\documentclass[a4paper,10pt,numbers=noenddot,twoside=on,headinclude=true,footinclude=false,headsepline=true]{scrartcl}

\usepackage[utf8]{inputenc}
\usepackage[T1]{fontenc}
\usepackage{textcomp}
\usepackage[english]{babel}
\usepackage{color}
\usepackage{amsmath}
\usepackage{amsopn}
\usepackage{amsbsy}
\usepackage[shortlabels]{enumitem}
\usepackage[plainpages=false,pdfpagelabels,pdfborderstyle={/S/U/W 1.2}]{hyperref}
\usepackage{graphics}
\usepackage{units}
\usepackage[cal=boondoxo,bb=boondox]{mathalfa}
\usepackage{dsfont}
\usepackage{csquotes}
\numberwithin{equation}{section} 
\numberwithin{table}{section} 
\numberwithin{figure}{section} 

\setlist[enumerate,1]{leftmargin=2.00em,itemsep=1pt,listparindent=15pt} % this only works with the package enumitem
\setlist[enumerate,2]{leftmargin=2.00em,itemsep=1pt,listparindent=15pt} % this only works with the package enumitem
\usepackage[amsmath,thmmarks,thref,hyperref]{ntheorem}

\theoremstyle{plain}
\newtheorem{theorem}{Theorem}[section]
\newtheorem{definition}[theorem]{Definition}
\newtheorem{lemma}[theorem]{Lemma}
\newtheorem{corollary}[theorem]{Corollary}

\newtheorem{proposition}[theorem]{Proposition}
\newtheorem{assumption}[theorem]{Assumption}

\theorembodyfont{\upshape}
\newtheorem{remark}[theorem]{Remark}

\theoremstyle{nonumberplain}

\theoremsymbol{\ensuremath{_\Box}}
\newtheorem{proof}{Proof}
\usepackage[scaled]{beramono}
\usepackage{helvet}
\usepackage[charter]{mathdesign} % math font has ugly \mathcals
\providecommand{\ie}{i.~e.~}
\providecommand{\eg}{e.~g.~}
\providecommand{\cf}{cf.~}

\providecommand{\R}{\mathbb{R}}

\providecommand{\C}{\mathbb{C}}

\renewcommand{\C}{\mathbb{C}}

\providecommand{\N}{\mathbb{N}}
\providecommand{\Z}{\mathbb{Z}}
\providecommand{\ii}{\mathrm{i}}
\providecommand{\e}{\mathrm{e}}

\providecommand{\Hil}{\mathcal{H}}
\providecommand{\eps}{\varepsilon}
\providecommand{\Cont}{\mathcal{C}}

\providecommand{\supp}{\mathrm{supp} \,}

\providecommand{\trace}{\mathrm{Tr} \,}

\providecommand{\dd}{\mathrm{d}}
\providecommand{\id}{\mathds{1}}

\providecommand{\Fourier}{\mathcal{F}}

\providecommand{\trace}{\mathrm{Tr}}

\providecommand{\abs}[1]{\left \lvert #1 \right \rvert}
\providecommand{\sabs}[1]{\lvert #1 \vert}
\providecommand{\babs}[1]{\bigl \lvert #1 \bigr \rvert}

\providecommand{\norm}[1]{\left \lVert #1 \right \rVert}
\providecommand{\snorm}[1]{\lVert #1 \rVert}
\providecommand{\bnorm}[1]{\bigl \lVert #1 \bigr \rVert}

\providecommand{\scpro}[2]{\left \langle #1 , #2 \right \rangle}
\providecommand{\sscpro}[2]{\langle #1 , #2 \rangle}
\providecommand{\bscpro}[2]{\bigl \langle #1 , #2 \bigr \rangle}
\providecommand{\Bscpro}[2]{\Bigl \langle #1 , #2 \Bigr \rangle}

\providecommand{\sopro}[2]{\vert #1 \rangle \langle #2 \vert}

\providecommand{\sexpval}[1]{\langle #1 \rangle}
\providecommand{\bexpval}[1]{\bigl \langle #1 \bigr \rangle}

\providecommand{\weyl}{\star}
\providecommand{\semisuper}{\bullet}
\providecommand{\super}{\sharp}

\providecommand{\op}{\mathrm{op}}
\providecommand{\Op}{\mathrm{Op}}

\providecommand{\Schwartz}{\mathcal{S}}
\providecommand{\bounded}{\mathcal{B}}

\providecommand{\ssMoyalSpace}{\mathfrak{m}^B(\Pspace)}

\providecommand{\pspace}{\Xi}
\providecommand{\Pspace}{\Xi^2}

\providecommand{\Xbf}{\mathbf{X}}
\providecommand{\Ybf}{\mathbf{Y}}

\providecommand{\ssMoyalSpace}{\mathfrak{m}^B(\Pspace)}

\title{A Proof of $\mathfrak{L}^2$-Boundedness for \\ Magnetic Pseudodifferential \\ Super Operators \\ via Matrix Representations \\ With Respect to Parseval Frames}
\author{Gihyun Lee${}^1$ \& Max Lein${}^2$}

\begin{document}

\maketitle
\vspace{-9mm}
\begin{center}
	${}^1$ Department of Mathematics: Analysis, Logic and Discrete Mathematics, Ghent University \linebreak
	Krijgslaan 281, S8 9000 Gent, Belgium \linebreak
	{\footnotesize \href{mailto:Gihyun.Lee@UGent.be}{\texttt{Gihyun.Lee@UGent.be}}}
	\medskip
	\\
	${}^2$ University of Postdam, Institute for Mathematics \linebreak
	Campus Golm, Building 9, Karl-Liebknecht-Str. 24-25, 14476 Potsdam OT Golm, Germany \linebreak
	{\footnotesize \href{mailto:lein2@uni-potsdam.de}{\texttt{lein2@uni-potsdam.de}}}
\end{center}
\begin{abstract}
	A fundamental result in pseudodifferential theory is the Calderón-Vaillancourt theorem, which states that a pseudodifferential operator defined from a Hörmander symbol of order $0$ defines a bounded operator on $L^2(\R^d)$. In this work we prove an analog for pseudodifferential \emph{super} operator, \ie operators acting on other operators, in the presence of magnetic fields. More precisely, we show that magnetic pseudodifferential super operators of order $0$ define bounded operators on the space of Hilbert-Schmidt operators $\mathfrak{L}^2 \bigl ( \mathcal{B} \bigl ( L^2(\R^d) \bigr ) \bigr )$. Our proof is inspired by the recent work of Cornean, Helffer and Purice \cite{Cornean_Helffer_Purice:boundedness_magnetic_PsiDOs_via_Gabor_frames:2022} and rests on a characterization of magnetic pseudodifferential super operators in terms of their “matrix elements” computed with respect to a Parseval frame. 
\end{abstract}

\newpage
\tableofcontents

%%% begin content %%% (fold)
%!TEX root = ./humble boundedness super PsiDOs.tex
\section{Introduction} % (fold)
\label{intro}
%
% CHANGED Add: non-commutative L^p spaces = quantum L^p spaces, etc. 
This work is a continuation of the authors' study of magnetic pseudodifferential \emph{super} operators \cite{Lee_Lein:magnetic_pseudodifferential_super_operators_basics:2022}, and we aim to prove boundedness criteria for them. An example of such a super operator would be the Liouville super operator
\begin{align*}
	\hat{L}^A(\hat{\rho}^A) = - \ii \, \bigl [ \op^A(h) , \hat{\rho}^A \bigr ] 
\end{align*}
or a more general Lindblad super operator acting on some trace class operator $\hat{\rho}^A \in \mathfrak{L}^1 \bigl ( \mathcal{B} \bigl ( L^2(\R^d) \bigr ) \bigr )$, where $\op^A(h)$ is the magnetic Weyl quantization \cite{Mantoiu_Purice:magnetic_Weyl_calculus:2004,Lein:progress_magWQ:2010,Iftimie_Mantoiu_Purice:magnetic_psido:2006} of a suitable function such as $h(x,\xi) \equiv h(X) = \xi^2 + V(x)$. Readers unfamiliar with \emph{magnetic} pseudodifferential theory can ignore the presence of the magnetic vector potential $A$ for the moment; we will give the relevant definitions in Section~\ref{setting} below. 

Our work \cite{Lee_Lein:magnetic_pseudodifferential_super_operators_basics:2022} explains that such a Liouville super operator can be viewed as the quantization $\hat{L}^A = \Op^A(L)$ with respect to the function 
\begin{align*}
	L(X_L,X_R) = - \ii \, \bigl ( h(X_L) - h(X_R) \bigr ) 
	, 
	&&
	X_L , X_R \in \pspace 
	. 
\end{align*}
Both, left ($L$) and right ($R$) variables are elements of phase space $\pspace := T^*\R^d \cong \R^d \times \R^d$, which we view as a cotangent bundle endowed with the magnetic symplectic form (see \eg \cite[Sections~2.1.1.1 and 3.6]{Lein:progress_magWQ:2010}). The indices $L$ and $R$ indicate whether after quantization the operator will act from the left or the right. 

That poses the question we wish to answer here: for generic functions $L$ belonging to some Hörmander class, under what conditions on $L$ does the associated magnetic pseudodifferential \emph{super} operator define a bounded operator 
\begin{align} \label{intro:eqn:L1_boundedess}
	\Op^A(L) : \mathfrak{L}^1 \bigl ( \mathcal{B} \bigl ( L^2(\R^d) \bigr ) \bigr ) \longrightarrow \mathfrak{L}^1 \bigl ( \mathcal{B} \bigl ( L^2(\R^d) \bigr ) \bigr ) 
\end{align}
between trace class operators? More generally, when is 
\begin{align} 
\label{intro:eqn:Lp_boundedness}
	\Op^A(L) \in \mathcal{B} \Bigl ( \mathfrak{L}^p \bigl ( \mathcal{B} \bigl ( L^2(\R^d) \bigr ) \bigr ) \Bigr ) 
\end{align}
a bounded operator between $p$-Schatten classes where $1 \leq p < \infty$? Or we could pose the same question for the space of bounded operators $\mathfrak{L}^{\infty} \bigl ( \mathcal{B} \bigl ( L^2(\R^d) \bigr ) \bigr ) = \mathcal{B} \bigl ( L^2(\R^d) \bigr )$ to itself, which corresponds to $p = \infty$. 

As the notation suggests, these operator spaces are particular cases of non-commutative $\mathfrak{L}^p$ spaces \cite{Terp:noncommutative_Lp_spaces:1981,Takesaki:operator_algebras_2:2003}, where the trace plays the role of integration with respect to a measure. Since the pseudodifferential context involves additional structures derived from twisted crossed product $C^*$-algebras \cite{Mantoiu_Purice_Richard:twisted_X_products:2004}, some subcommunities refer to non-commutative $\mathfrak{L}^p$-spaces equipped with these extra structures as \emph{quantum} $\mathfrak{L}^p$ spaces (\cf \cite[Section~3]{Antoniou_Majewski_Suchanecki:Liouville_evolution_Koopman_formalism:2002}). Likewise, imposing additional regularity gives rise to \eg \emph{non-commutative} or \emph{quantum} Sobolev spaces (see, \eg \cite{Lafleche:quantum_Sobolev_inequalities:2022}). 

Our main result answers the question of boundedness~\eqref{intro:eqn:Lp_boundedness} for $p = 2$, which should be seen as a super operator analog of the Calderón-Vaillancourt theorem from psuedodifferential theory (\cf \eg \cite[Theorem~3.1]{Iftimie_Mantoiu_Purice:magnetic_psido:2006}): 
\begin{theorem}\label{intro:thm:boundedness_super_PsiDOs}
	Suppose the components $B_{jk} \in \Cont^{\infty}_{\mathrm{b}}(\R^d)$ of the magnetic field $B = \sum_{j , k = 1}^d B_{jk} \, \dd x_j \wedge \dd x_k$ are smooth, bounded and possess bounded derivatives to any order. Moreover, let $A \in \Cont^{\infty}_{\mathrm{pol}}(\R^d,\R^d)$ be a smooth, polynomially bounded vector potential with $\dd A = B$. 
	
	Then any Hörmander symbol $L \in S^m_{\rho,0}(\Pspace)$ (\cf Definition~\ref{setting:defn:Hoermander_symbol_classes}) of order $m \leq 0$ and type $(\rho,0)$ where $0 \leq \rho \leq 1$ defines a bounded operator
	\begin{align*}
		\Op^A(L) \in \mathcal{B} \Bigl ( \mathfrak{L}^2 \bigl ( \mathcal{B} \bigl ( L^2(\R^d) \bigr ) \bigr ) \Bigr ) 
		. 
	\end{align*}
	%
	% \textcolor{blue}{Just to be sure: is our proof covers the case $\rho>0$?}
\end{theorem}
In fact, we are not only able to treat regular Hörmander classes, but also Hörmander classes where we allow the growth in left and right momentum variables $\xi_{L,R}$ to be different (\cf Theorem~\ref{boundedness_super_PsiDOs:thm:boundedness_super_PsiDOs} and Corollary~\ref{boundedness_super_PsiDOs:cor:boundedness_super_PsiDOs} for details).

\subsection{Our method of choice: representing magnetic pseudodifferential super operators as “infinite matrices”} % (fold)
\label{intro:infinite_matrices}
There are several established strategies to prove the boundedness of pseudodifferential operators. One of them (\eg implemented for magnetic pseudodifferential operators as the proof of \cite[Theorem~3.1]{Iftimie_Mantoiu_Purice:magnetic_psido:2006}) slices up the operator into pieces that are localized in a grid cell, estimates each localized operator with oscillatory integral techniques and controls their sum, the original operator, via the Stein-Cotlar-Knapp Lemma. 
% 
% Second path to proof
%  - Serge Alinhac and Patrick Gérard, Pseudo-differential Operators and the Nash-Moser Theorem (2007).
%    The proof can be found in page 30.
%  - Xavier Saint Raymond, Elementary Introduction to the Theory of Pseudodifferential Operators (1991).
%    The proof can be found in pages 53-55.
% - Schur's test 
We could alternatively follow the strategy implemented by Alinhac and Gérard (\cf \cite[pp.~29–31]{Alinhac_Gerard:PsiDOs_Nash_Moser_Theorem:2007}) and Saint Raymond (\cf \cite[pp.~53–55]{Saint_Raymond:intro_pseudodifferential_theory:1991}), and base our proof on Schur's test; the latter gives an explicit criterion for the boundedness of operators in terms of their operator kernels. 

However, we will choose a different path inspired by two recent works by Cornean, Helffer and Purice \cite{Cornean_Helffer_Purice:simple_proof_Beals_criterion_magnetic_PsiDOs:2018,Cornean_Helffer_Purice:boundedness_magnetic_PsiDOs_via_Gabor_frames:2022}. The idea is to characterize magnetic pseudodifferential (super) operators through their “matrix elements” with respect to a Gabor frame, an approach which was pioneered by Cordero, Heil, Gröchenig and others (\cf \eg \cite{Heil_Groechenig:modulation_spaces_PsiDOs:1999,Goechenig:time_frequency_analysis_PsiDOs:2006,Goechenig_Rzeszotnik:Banach_algebras_PsiDOs_and_their_almost_diagonalization:2008}). A frame is a countable set of vectors that forms an “overcomplete basis” of a separable Hilbert space (\cf Section~\ref{setting:Parseval_frame} and references therein for a proper mathematical definition). Operators are therefore uniquely defined by their matrix elements. The modifier “Gabor” refers to the fact that in this case the frame is generated from a single vector and indexed by a lattice vector. The Gabor frame Cornean et al.\ use is also a tight, normalized (or \emph{Parseval}) frame of $L^2(\R^d)$, \ie a frame that in many respects acts like an orthonormal basis. 

There are multiple advantages to this approach: first of all, one only needs to use oscillatory integral techniques one time, namely in \cite[Theorem~3.1]{Cornean_Helffer_Purice:boundedness_magnetic_PsiDOs_via_Gabor_frames:2022} and our analogs, Theorem~\ref{characterization_symbols_matrix_elements_super_PsiDOs:thm:characterization_Hoermander_class_super_PsiDOs} and Corollary~\ref{characterization_symbols_matrix_elements_super_PsiDOs:cor:characterization_Hoermander_class_super_PsiDOs}. These results give a one-to-one characterization of magnetic pseudodifferential \emph{super} operators associated to a Hörmander symbol and their matrix elements. Proofs of subsequent theorems then only involve matrix elements, and a common theme is to impose conditions that ensure the existence of certain infinite sums over the matrix elements' indices. Conceptually and technically, this is a great simplification. Moreover, it avoids the tedium of proving the existence of certain oscillatory integrals and estimating relevant seminorms — which is further exacerbated when additional magnetic phase factors and variables are present (see \eg \cite[Appendix~B]{Lee_Lein:magnetic_pseudodifferential_super_operators_basics:2022}). 

The second main advantage is the ease of adaptability: one of the authors used this matrix element point of view to extend commutator criteria to magnetic pseudodifferential operators associated to operator-valued and equivariant operator-valued symbols (specifically Theorems~3.4.23 and 4.3.1 in \cite{DeNittis_Lein_Seri:equivariant_PsiDOs:2022}). A lot of more advanced results crucially depend on the given Fréchet topology. And because the Fréchet spaces involved are usually \emph{not} nuclear, it is not possible to \emph{easily} adapt existing proofs in a straightforward and transparent manner. In fact, this is true here, too: the relevant Hörmander classes (\cf Definition~\ref{setting:defn:double_Hoermander_symbol_classes}) are over $\Pspace = \pspace \times \pspace := T^* \R^d \times T^* \R^d$, and it would be tempting to split 
\begin{align*}
	\mbox{“$S^{m_L,m_R}_{\rho,\delta}(\Pspace) \cong S^{m_L}_{\rho,\delta}(\pspace) \otimes S^{m_R}_{\rho,\delta}(\pspace)$”} 
	&&
	\mbox{(incorrect equation)}
\end{align*}
as a tensor product, and then invoke \cite[Theorem~3.1]{Cornean_Helffer_Purice:boundedness_magnetic_PsiDOs_via_Gabor_frames:2022} to obtain a proof of Theorem~\ref{characterization_symbols_matrix_elements_super_PsiDOs:thm:characterization_Hoermander_class_super_PsiDOs}. Unfortunately, \emph{the above equation is false.} Hörmander symbol classes are non-nuclear Fréchet spaces \cite{Witt:weak_topology_symbol_spaces:1997}, and consequently, the tensor product on the right is ill-defined: there exist \emph{at least} two topologies with respect to which we can complete the algebraic tensor product. And neither of these two gives the space on the left (\cf \eg Theorem~3.2, Remark~3.3, and Proposition~4.4 in \cite{Witt:weak_topology_symbol_spaces:1997}). Yet at the end of the day, the result \emph{is} as if we could treat the Hörmander class on the left as a tensor product. The characterization via matrix elements \emph{does} neatly extend to these more general cases as \cite{DeNittis_Lein_Seri:equivariant_PsiDOs:2022} and the present work show. 
% subsection representing_magnetic_pseudodifferential_operators_as_infinite_matrices (end)

\subsection{The difficulty in proving $\mathfrak{L}^p$ boundedness for $p \neq 2$} % (fold)
\label{sub:the_difficulty_in_proving_mathfrak_l_p_boundedness_for_p_neq_2}
We framed our work by looking at Liouville and Lindblad super operators, which typically act on (dense subspaces of) $\mathfrak{L}^1 \bigl ( \mathcal{B} \bigl ( L^2(\R^d) \bigr ) \bigr )$. Unfortunately, we had to exclude the highly desirable cases $p = 1$ (trace-class operators) and $p = \infty$ (bounded operators). For the benefit of the interested reader, we will explain a few of the reasons and put our results into context with the state-of-the-art. We intend to return to this subject in the future.

\subsubsection{The lack of a simple characterization of trace-class and bounded operators in terms of their matrix elements} % (fold)
\label{ssub:a_lack_of_a_simple_characterization_of_trace_class_and_bounded_operators_in_terms_of_their_matrix_elements}
The key ingredient in our approach is to represent super operators as well as the operators they act on as “infinite matrices”. Each of the matrix elements is defined with respect to a Parseval frame $\{ \mathcal{G}_{\alpha} \}_{\alpha \in \Gamma}$ for a lattice $\Gamma$, and properties such as $\hat{g} \in \mathfrak{L}^p \bigl ( \mathcal{B}(\Hil) \bigr )$ should translate to properties of 
\begin{align*}
	\hat{g}_{\alpha , \beta} := \bscpro{\mathcal{G}_{\alpha}}{\hat{g} \, \mathcal{G}_{\beta}} 
	, 
\end{align*}
where \eg $\Hil = L^2(\R^d)$ and $\alpha$ and $\beta$ take values in some countably infinite index set $\Gamma$. The application of a super operator $\hat{F}$ on an operator $\hat{g}$ can now be described as a matrix product\footnote{The position of the indices to be summed over may seem unusual, but we will explain our choice in Section~\ref{abstract_matrix_representation}. } 
\begin{align*}
	(\hat{F} \hat{g})_{\alpha , \beta} = \sum_{\alpha' , \beta' \in \Gamma} \hat{F}_{\alpha , \beta' , \alpha' , \beta} \, \hat{g}_{\beta' , \alpha'}
	. 
\end{align*}
Hence, proving boundedness of the super operator $\hat{F} \in \mathcal{B} \bigl ( \mathfrak{L}^p \bigl ( \mathcal{B}(\Hil) \bigr ) \bigr )$ is reduced to arguments ensuring the existence of the above sum in a specific sense. However, that requires a characterization of the matrix elements of $p$-Schatten class operators. 

When $p = 2$ we may identify the Hilbert space of Hilbert-Schmidt operators on $\Hil$ with the Hilbert space $\Hil \otimes \Hil \cong \ell^2(\Gamma) \otimes \ell^2(\Gamma)$ in a canonical fashion. Unfortunately, for all other values of $p$ we are not aware of such a simple characterization. Specifically, results for $p = 1$ and $p = \infty$ are particularly desirable, not only because they are the most relevant in applications, but also because they form the two end points in a Riesz-Thorin interpolation argument (\cf \eg Theorems~2.9 and 2.10 as well as Remark~1 on p.~23 in \cite{Simon:trace_ideals_applications:2005} or \cite[Theorem~13.1]{Gohberg_Krein:linear_nonselfadjoint_operators:1969}). 

For $p < \infty$ the best results we are aware of are Theorems~A and B in \cite{Bingyang_Khoi_Zhu:frames_Schatten_classes:2015}. Simply put, they characterize the cases $2 \leq p < \infty$ and $0 < p \leq 2$, respectively, in terms of the $p$-summability of $\snorm{\hat{g} \mathcal{G}_{\alpha}}$. Unfortunately, Theorem~B which covers trace-class operators states that we need $p$-summability with respect to \emph{some} frame; that is, the frame will in general depend on the operator under consideration. 

The case $p = \infty$ has to be treated separately. Fortunately, if all we care about is proving the boundedness of magnetic pseudodifferential \emph{super} operators, we do not need a one-to-one characterization of bounded operators in terms of their matrix elements. Still, such a characterization exists and can be found in \eg Chapter~1 of the recent monograph \cite{deMalafosse_Malkowsky_Rakocevic:operators_sequence_spaces:2021}, specifically Lemmas~1.3 and 1.4 as well as Crone's Theorem~1.24. Unfortunately, none of these conditions translate to simple conditions on the matrix elements. What is more, the existing results mentioned below indicate that looking at 
\begin{align*}
	\Op^A(F) : \mathfrak{L}^p \bigl ( \mathcal{B} \bigl ( L^2(\R^d) \bigr ) \bigr ) \longrightarrow \mathfrak{L}^p \bigl ( \mathcal{B} \bigl ( L^2(\R^d) \bigr ) \bigr ) 
\end{align*}
might not be the way to go for the cases $p = 1$ and $p = \infty$. 
% subsubsection the_lack_of_a_simple_characterization_of_trace_class_and_bounded_operators_in_terms_of_their_matrix_elements (end)

\subsubsection{The state-of-the-art for some simple cases} % (fold)
\label{ssub:the_state_of_the_art_for_some_simple_cases}
There \emph{is} a partial result in this direction, \eg \cite{GonzalezPerez_Junge_Parcet:singular_integral_operators_quantum_Euclidean_spaces:2021}. The space $\mathcal{B} \bigl ( L^2(\R^d) \bigr )$ is a special case of quantum Euclidean spaces studied in \cite{GonzalezPerez_Junge_Parcet:singular_integral_operators_quantum_Euclidean_spaces:2021} (\cf loc.\ cit.\ Remark~1.3). Therefore, the pseudodifferential theory developed in \cite{GonzalezPerez_Junge_Parcet:singular_integral_operators_quantum_Euclidean_spaces:2021} pertains to super operators that multiply with a pseudodifferential operator \emph{from one side}, \ie (non-magnetic) super operators of the form 
\begin{align*}
	\Op^{A = 0} \bigl ( f_L \otimes 1 \bigr ) = \op^{A = 0}(f_L) \otimes \id : \hat{g} \mapsto \op^{A = 0}(f_L) \, \hat{g} 
\end{align*}
when we view them as operators acting from the left.\footnote{As Gonz\'alez-P\'erez, Junge and Parcet have noted in \cite[Remark~2.5]{GonzalezPerez_Junge_Parcet:singular_integral_operators_quantum_Euclidean_spaces:2021}, we could have equivalently considered operators acting from the right. The two cases are connected via the adjoint operation.} The calculus for magnetic pseudodifferential super operators we have developed in  \cite{Lee_Lein:magnetic_pseudodifferential_super_operators_basics:2022} applies to much more general super operators, which need \emph{not} be “product super operators” of the form 
\begin{align*}
	\Op^A \bigl ( f_L \otimes f_R \bigr ) : \hat{g}^A \mapsto \op^A(f_L) \, \hat{g}^A \, \op^A(f_R)
\end{align*}
and incorporate magnetic fields in a covariant fashion. Still, this simplified case can give us guidance on what results we should expect to hold true. The proof of $L^p$-boundedness of pseudodifferential operators on $\R^d$ for $1 < p < \infty$ can be found in \cite{Fefferman:Lp_boundedness_PsiDOs:1973,Stein:harmonic_analysis:1993}. Gonz\'alez-P\'erez, Junge and Parcet extended this $\mathfrak{L}^p$-boundedness result to the setting of quantum Euclidean spaces \cite{GonzalezPerez_Junge_Parcet:singular_integral_operators_quantum_Euclidean_spaces:2021}. 

Importantly, these authors also have partial results for the cases $p = 1$ and $p = \infty$. In both cases, either initial or target spaces is not a $\mathfrak{L}^p$ space. Instead, Gonz\'alez-P\'erez et al.\ obtain boundedness (\cf \cite[Theorem~2.18]{GonzalezPerez_Junge_Parcet:singular_integral_operators_quantum_Euclidean_spaces:2021}) if we view these super operators as operators acting between 
\begin{align}
	\Op^{A = 0} \bigl ( f_L \otimes 1 \bigr ) : \text{H}^1 \bigl ( \mathcal{B}(\Hil) \bigr ) \rightarrow \mathfrak{L}^1 \bigl( \mathcal{B} (\Hil) \bigr ) 
	, 
	\label{intro:eqn:left_super_operator_boundedness_H1_L1}
	\\
	\Op^{A = 0} \bigl ( f_L \otimes 1 \bigr ) : \mathcal{B}(\Hil) \longrightarrow \text{BMO} \bigl ( \mathcal{B} (\Hil) \bigr ) 
	,
	\label{intro:eqn:left_super_operator_boundedness_Linfty_BMO}
\end{align}
where $\text{H}^1 \bigl ( \mathcal{B} (\Hil) \bigr )$ and $\text{BMO} \bigl ( \mathcal{B} (\Hil) \bigr )$ are the non-commutative versions of the Hardy space and the space of functions of bounded mean oscillation. 

This might suggest that at least in the context of pseudodifferential super operators it is more natural to look at generalizations of \eqref{intro:eqn:left_super_operator_boundedness_H1_L1} and \eqref{intro:eqn:left_super_operator_boundedness_Linfty_BMO} to more general functions $F \neq f_L \otimes 1$ such as Hörmander symbols from Definitions~\ref{setting:defn:Hoermander_symbol_classes} and \ref{setting:defn:double_Hoermander_symbol_classes}. 
% subsubsection the_state_of_the_art_for_some_simple_cases (end)
% subsection the_difficulty_in_proving_mathfrak_l_p_boundedness_for_p_neq_2 (end)

\subsection*{Outline} % (fold)
\label{intro:outline}
Apart from the introduction, this paper consists of four sections: in Section~\ref{setting} we detail our setting and give the necessary definitions and assumptions. Then in Section~\ref{abstract_matrix_representation} we characterize certain classes of (super) operators in terms of their matrix elements. The only time we will need oscillatory integral techniques is in Section~\ref{characterization_symbols_matrix_elements_super_PsiDOs}, where we characterize magnetic pseudodifferential super operators in terms of their matrix elements. Lastly, proofs of the main result, Theorem~\ref{intro:thm:boundedness_super_PsiDOs}, and its more general statements, Theorem~\ref{boundedness_super_PsiDOs:thm:boundedness_super_PsiDOs} as well as Corollaries~\ref{boundedness_super_PsiDOs:cor:boundedness_super_PsiDOs} and \ref{boundedness_super_PsiDOs:cor:boundedness_super_PsiDOs_magnetic_Sobolev_spaces_double_Hoermander_symbols}, are given in Section~\ref{boundedness_super_PsiDOs}. 
% subsection outline (end)

\subsection*{Data availability} % (fold)
\label{intro:data}
No datasets were generated or analyzed, because this work is based on a purely mathematical approach. Therefore, data sharing is not applicable to this article.
% subsection data availability (end)

\subsection*{Acknowledgements} % (fold)
\label{intro:acknowledgements}
G.\ L.\ was supported by the Max Planck Institute for Mathematics, Bonn, the FWO Odysseus 1 grant G.0H94.18N: Analysis and Partial Differential Equations, and the Methusalem programme of the Ghent University Special Research Fund (BOF) (Grant number 01M01021).

M.\ L.\ acknowledges the JSPS for the support of this project through a Kiban~C grant (No. 20K03761). He would also like to thank Max Planck Institute for Mathematics for its kind hospitality. 

Both authors would like to express their gratitude towards Horia Cornean and Radu Purice for very helpful remarks and stimulating discussions. 
% subsection acknowledgements (end)
% section introduction (end)
%!TEX root = ./humble boundedness super PsiDOs.tex
\section{Setting and fundamental definitions} % (fold)
\label{setting}
In this section we will fix some notation, give some basic definitions and recount basic facts.

\subsection{The relevant Parseval frame} % (fold)
\label{setting:Parseval_frame}
Let $\Hil$ be a separable, infinite-dimensional Hilbert space endowed with an inner product $\scpro{\, \cdot \,}{\, \cdot \,}$. A frame generalizes the notion of a basis; it is essentially an “overcomplete basis” of $\Hil$, which still allows for norm estimates. More precisely, a sequence $\{ \mathcal{G}_{\alpha} \}_{\alpha \in \Gamma}$ in $\Hil$ indexed by a countable set $\Gamma$ is a frame if and only if there exist two constants $0 < C_- \leq C_+ < \infty$ such that
\begin{align}
	C_- \, \snorm{\psi}^2 \leq \sum_{\alpha \in \Gamma} \babs{\scpro{\mathcal{G}_{\alpha}}{\psi}}^2 \leq C_+ \, \snorm{\psi}^2 
	\label{setting:eqn:definition_frame_lower_upper_bound}
\end{align}
holds for all $\psi \in \Hil$. A frame is called tight if and only if we may choose $C_- = C_+$; a tight frame for which $C_+ = 1 = C_-$ is called a normalized tight or \emph{Parseval frame}. Parseval frames act much like orthonormal basis in that \eg we may expand any vector 
\begin{align*}
	\psi = \sum_{\alpha \in \Gamma} \bscpro{\mathcal{G}_{\alpha}}{\psi} \, \mathcal{G}_{\alpha} 
\end{align*}
in terms of the elements of the Parseval frame. In our convention, the first argument of the scalar product is antilinear and the second one is linear. 

% CHANGED Add the following bits: identification of vector with collection of coefficients; \ell^2_{\Hil}
For Parseval frames $\{ G_{\alpha} \}_{\alpha \in \Gamma}$ we can identify each vector $\varphi \in \Hil$ with its collection of coefficients 
\begin{align*}
	(\varphi_{\alpha})_{\alpha \in \Gamma} := \bigl ( \sscpro{\mathcal{G}_{\alpha}}{\varphi} \bigr )_{\alpha \in \Gamma} 
	\in \ell^2(\Gamma) 
	. 
\end{align*}
The Parseval properties can be rephrased as saying that the map 
\begin{align*}
	\mathfrak{u}_2 : \Hil \longrightarrow \ell^2(\Gamma) 
	, 
	\quad 
	\varphi \mapsto (\varphi_{\alpha})_{\alpha \in \Gamma} 
\end{align*}
is a linear isometry between Banach spaces; when $\{ \mathcal{G}_{\alpha} \}_{\alpha \in \Gamma}$ is only a Parseval frame, but not an orthonormal basis, $\mathfrak{u}_2$ fails to be onto, though. Consequently, for some arguments, we may need to restrict $\mathfrak{u}_2$ to its range 
\begin{align*}
	\ell^2_{\Hil}(\Gamma) := \mathfrak{u}_2(\Hil) 
	\subseteq \ell^2(\Gamma) 
	. 
\end{align*}
Frequently, frames are generated from a single vector $\chi \in \Hil$, and we may set 
\begin{align*}
	\mathcal{G}_{\alpha} := \pi(\alpha) \, \chi 
	, 
\end{align*}
where $\Gamma \cong \Z^k$ for some $k \in \N$ and $\pi : \Gamma \longrightarrow \mathcal{U}(\Hil)$ is a unitary-operator-valued map. Such frames are called \emph{Gabor frames}. For an introduction to the theory of frames, we refer to \cite{Christensen:introduction_frames:2003}. 
\medskip

\noindent
The idea to use Gabor frames to study pseudodifferential operators goes back to the late 1990s and early 2000s with the pioneering works of Gröchenig, Cordero, Heil and co-workers \cite{Heil_Groechenig:modulation_spaces_PsiDOs:1999,Cordero_Groechenig:time_frequency_analysis_localization_operators:2003,Cordero_Groechenig:necessary_conditions_Schatten_class_localization_operators:2005}. More recently, it has been applied to \emph{magnetic} pseudodifferential operators by Cornean, Helffer and Purice \cite{Cornean_Helffer_Purice:simple_proof_Beals_criterion_magnetic_PsiDOs:2018,Cornean_Helffer_Purice:boundedness_magnetic_PsiDOs_via_Gabor_frames:2022}, and in spirit and substance, our work borrows its main ideas from them and applies them in new ways. 

We will use a Gabor frame, which is simultaneously a Parseval frame. It starts with a function $\chi \in \Cont^{\infty}_{\mathrm{c}}\bigl ( \R^d , [0,\infty) \bigr )$ and a lattice $\Gamma \cong \Z^d$. We require that $\chi$ satisfies 
\begin{enumerate}[(a)]
	\item $\supp \chi \subseteq (-1,+1)^d$ and 
	\item $\sum_{\gamma \in \Gamma} \chi(x - \gamma)^2 = 1$ holds for all $x \in \R^d$. 
\end{enumerate}
Furthermore, we will need to introduce the \emph{dual lattice}
\begin{align*}
	\Gamma^* := \Bigl \{ \gamma^* \in \R^d \; \; \big \vert \; \; \gamma^* \cdot \gamma \in 2\pi \Z \; \forall \gamma \in \Gamma \Bigr \} 
	, 
\end{align*}
which is again isomorphic to $\Z^d$. Lastly, we need to introduce a magnetic field $B = \dd A$ and a vector potential $A$ representing it. One may think of magnetic fields as smooth, closed two-forms $B = \sum_{j,k = 1}^d B_{jk} \, \dd x_j \wedge \dd x_k$ on $\R^d$; we shall always tacitly identify $B$ with the matrix-valued function $(B_{jk})_{1 \leq j , k \leq d}$. A magnetic vector \emph{potential} $A = \sum_{j = 1}^d A_j \, \dd x_j$ is said to represent a magnetic field $B$ if and only if $B = \dd A$; whenever convenient we will think of the one-form $A$ as the vector-valued function $(A_1 , \ldots , A_d)$. Throughout this article, we will make the following assumptions.
\begin{assumption}[Magnetic fields and vector potentials]\label{setting:assumption:magnetic_field}
	\begin{enumerate}[(a)]
		\item The components $B_{jk}, 1 \leq j , k , \leq d$, of the magnetic field $B$ are of class $\Cont^{\infty}_{\mathrm{b}}$, \ie bounded, smooth and with bounded derivatives to any order. 
		\item A vector potential $A \in \Cont^{\infty}_{\mathrm{pol}}(\R^d,\R^d)$ representing such magnetic fields $\dd A = B$ has polynomially bounded, smooth components whose derivatives are all polynomially bounded as well. 
	\end{enumerate}
\end{assumption}
The exponential of the phase 
\begin{align*}
	\Lambda^A(x,y) := \e^{- \ii \int_{[x,y]} A} 
\end{align*}
obtained by integrating the vector potential along the line segment $[x,y]$ will enter into the definition of our Gabor frame, 
\begin{align}
	\mathcal{G}^A_{\alpha,\alpha^*}(x) := (2\pi)^{- \nicefrac{d}{2}} \Lambda^A(x,\alpha) \, \e^{+ \frac{\ii}{2\pi} \alpha^* \cdot (x - \alpha)} \, \chi(x - \alpha) 
	\in \Schwartz(\R^d)
	. 
	\label{setting:eqn:definition_Gabor_frame}
\end{align}
The presence of the factor $\nicefrac{1}{2\pi}$ stems from the embedding $(-1,+1)^d \subset (-\pi,+\pi)^d$. Functions with compact support inside the cube  $(-1,+1)^d$ are then replaced by their $2\pi$-periodization (\cf the explanations surrounding equations~(2.2) and (2.3) in \cite{Cornean_Helffer_Purice:boundedness_magnetic_PsiDOs_via_Gabor_frames:2022}). 

We will restate \cite[Proposition~2.2]{Cornean_Helffer_Purice:boundedness_magnetic_PsiDOs_via_Gabor_frames:2022} for convenience: 
\begin{proposition}
	\begin{enumerate}[(1)]
		\item $\bigl \{ \mathcal{G}^A_{\alpha,\alpha^*} \bigr \}_{(\alpha,\alpha^*) \in \Gamma \times \Gamma^*}$ is a Parseval frame. 
		\item We may expand any $\psi \in L^2(\R^d)$ in terms of the Gabor frame as 
		\begin{align*}
			\psi &= \sum_{(\alpha,\alpha^*) \in \Gamma \times \Gamma^*} \bscpro{\mathcal{G}^A_{\alpha,\alpha^*}}{\psi} \, \mathcal{G}^A_{\alpha,\alpha^*} 
			, 
		\end{align*}
		where the above series converges in the topology of $L^2(\R^d)$. 
	\end{enumerate}
\end{proposition}
Item~(2) will allow us to identify operators $\hat{f}$ on $L^2(\R^d)$ with their collection of matrix elements $\hat{f}^A_{\alpha,\alpha^*,\beta,\beta^*} := \bscpro{\mathcal{G}^A_{\alpha,\alpha^*}}{\hat{f} \, \mathcal{G}^A_{\beta,\beta^*}}$. 
% subsection the_relevant_parseval_frame (end)

\subsection{Magnetic Weyl calculus}
\label{setting:magnetic_PsiDOs}
In this subsection, we will review the construction and main properties of magnetic Weyl calculus. For a more comprehensive account, we refer the readers to \eg \cite{Mantoiu_Purice:magnetic_Weyl_calculus:2004,Iftimie_Mantoiu_Purice:magnetic_psido:2006,Lein:progress_magWQ:2010}.

In what follows, we denote points on the phase space $\pspace := T^*\R^d \cong \R^d \times \R^d$ by $X = (x,\xi)$, $Y = (y,\eta)$ and $Z = (z,\zeta)$ with position variables $x , y , z \in \R^d$ and momentum variables $\xi , \eta , \zeta \in \R^d$.

The basic building blocks of magnetic Weyl calculus are position operators $Q = (Q_1,\ldots,Q_d)$ and kinetic momentum operators $P^A = (P_1^A,\ldots,P_d^A)$. These are self-adjoint unbounded operators on $L^2(\R^d)$ subject to the commutation relations
\begin{align} \label{setting:eqn:CCR}
	\ii \, [Q_j,Q_k] = 0 , \qquad \ii \, [P_j^A,P_k^A] = \eps \, \lambda \, B_{jk}(Q) , \qquad \ii \, [P_j^A,Q_k] = \eps \, \delta_{jk} 
	, 
	&&
	1 \leq j , k \leq d 
	.
\end{align}
Here $\eps$ is a semiclassical parameter and $\lambda \leq 1 $ is another parameter quantifying the coupling of the electric charge to the magnetic field. Both of them play an important role in the study of asymptotic expansions and semiclassical limits \cite{Lein:two_parameter_asymptotics:2008,Fuerst_Lein:scaling_limits_Dirac:2008}. 
\begin{remark}
	For the sake of brevity, we will set $\lambda = 1 = \eps$. We wish to emphasize that all of our subsequent results hold when the small parameters are restored. 
\end{remark}
We can put the commutation relations~\eqref{setting:eqn:CCR} into the framework of bounded operators by using the magnetic Weyl system defined by
\begin{align*}
	w^A(X) := \e^{-\ii\sigma(X,(Q,P^A))} = \e^{+ \ii (x \cdot P^A - \xi \cdot Q)} 
	, 
	&&
	X \in \pspace 
	,
\end{align*}
where we have used the symplectic form $\sigma(X,Y) := \xi \cdot y - x \cdot \eta$ on $\pspace$. The symplectic form also appears as the phase factor in the \emph{symplectic Fourier transform}
\begin{align*}
	(\Fourier_\sigma f)(X) := \frac{1}{(2\pi)^d} \int_{\pspace} \dd Y \, \e^{+\ii\sigma(X,Y)} \, f(Y) 
\end{align*}
that is initially defined for Schwartz functions, but its definition extends to tempered distributions. Compared to the standard Fourier transform, it has the added benefit of being its own inverse, $\Fourier_{\sigma}^2 = \id$. 

The magnetic Weyl pseudodifferential operator associated with $f$ is defined by
\begin{align}
	\op^A(f) := \frac{1}{(2\pi)^d} \int_\pspace \dd X \, (\Fourier_\sigma f)(X) \, w^A(X) 
	.
	\label{setting:eqn:magnetic_Weyl_quantization}
\end{align}
The magnetic Weyl quantization map $f \mapsto \op^A(f)$ yields a continuous linear operator $\op^A$ from $\Schwartz(\pspace)$ to $\mathcal{L} \bigl (\Schwartz'(\R^d) \, , \, \Schwartz(\R^d) \bigr )$; it extends to a topological vector space isomorphism from $\Schwartz'(\pspace)$ to $\mathcal{L} \bigl ( \Schwartz(\R^d) \, , \, \Schwartz'(\R^d) \bigr )$. In particular, as a Hörmander symbol space $S_{\rho,\delta}^m(\pspace)$ of order $m \in \R$, defined below, is included in $\Schwartz'(\pspace)$, and each symbol $f\in S_{\rho,\delta}^m(\pspace)$ yields a continuous linear operator from $\Schwartz(\pspace)$ to $\Schwartz'(\pspace)$. In addition, the magnetic Weyl quantization map $f \mapsto \op^A(f)$ also gives rise to a unitary map from $L^2(\pspace)$ to the space of Hilbert-Schmidt operators $\mathfrak{L}^2 \bigl ( \mathcal{B} \bigl ( L^2(\R^d) \bigr ) \bigr )$.
\begin{definition}[Hörmander symbol classes]\label{setting:defn:Hoermander_symbol_classes}
	Let $0 \leq \delta \leq \rho \leq 1$, $\delta<1$ and $m \in \R$. The space $S_{\rho,\delta}^m(\pspace)$ consists of smooth functions $f : \pspace \longrightarrow \C$ such that, for all $a,\alpha\in\N_0^d$ there exists $C_{a\alpha}>0$ such that
	\begin{align}
		\babs{\partial_x^a \partial_\xi^\alpha f(X)} \leq C_{a\alpha} \, \sexpval{\xi}^{m - \abs{\alpha} \rho + \abs{a} \delta} 
		&&
		\forall X \in \pspace 
		.
		\label{setting:eqn:Hoermander_symbol_estimates}
	\end{align}
\end{definition}
The smallest constants $C_{a\alpha}$ satisfying the estimates~\eqref{setting:eqn:Hoermander_symbol_estimates} are the seminorms
\begin{align*}
	\norm{f}_{m,a\alpha} := \sup_{X \in \pspace} \Bigl ( \sexpval{\xi}^{-m + \abs{\alpha} \rho - \abs{a} \delta} \babs{\partial_x^a \partial_\xi^\alpha f(X)} \Bigr ) 
	,
\end{align*}
where $\sexpval{\xi} := \sqrt{1 + \sabs{\xi}^2}$ is the Japanese bracket. We endow $S_{\rho,\delta}^m(\pspace)$ with the Fréchet space topology generated by the seminorms $\norm{\cdot}_{m,a\alpha}$, where $a$ and $\alpha$ range over all multi indices in $\N_0^d$.

The magnetic Weyl product $f\weyl^B g$ of two symbols or distributions $f$ and $g$ is defined via the composition of the pseudodifferential operators associated with them, 
\begin{align*}
	\op^A \bigl ( f \weyl^B g \bigr ) := \op^A(f) \, \op^A(g) 
	.
\end{align*}
There exist multiple equivalent explicit formulas for $\weyl^B$, although we will not need them in this work. 

This product pulls back the operator product to the level of functions or distributions on phase space. It gives rise to a continuous bilinear map between Hörmander classes (\cf \eg \cite[Theorem~2.6]{Iftimie_Mantoiu_Purice:magnetic_psido:2006}), 
\begin{align*}
	\weyl^B: S_{\rho,\delta}^{m_1}(\pspace)\times S_{\rho,\delta}^{m_2}(\pspace) \longrightarrow S_{\rho,\delta}^{m_1+m_2}(\pspace) 
	.
\end{align*}
Another standard result of magnetic Weyl calculus is the $L^2$-boundedness: magnetic Weyl quantization $\op^A$ gives rise to a continuous linear map from $S_{0,0}^0(\pspace)$ to $\bounded \bigl ( L^2(\R^d) \bigr )$. This magnetic version of the Calderón-Vaillancourt theorem was first proved by Iftimie, Mantoiu and Purice in~\cite{Iftimie_Mantoiu_Purice:magnetic_psido:2006}. An alternative approach using the magnetic Gabor frame~\eqref{setting:eqn:definition_Gabor_frame}, the approach which we will also utilize in this article to prove the boundedness of magnetic pseudodifferential \emph{super} operators, can be found in~\cite{Cornean_Helffer_Purice:boundedness_magnetic_PsiDOs_via_Gabor_frames:2022}.
% subsection magnetic_PsiDOs (end)

\subsection{Magnetic pseudodifferential super operator calculus}
\label{setting:magnetic_PsiD_super_Os}
The main motivation for introducing the calculus of magnetic Weyl pseudodifferential super operators in \cite{Lee_Lein:magnetic_pseudodifferential_super_operators_basics:2022} was to study generalizations of super operators of the form 
\begin{align*}
	\hat{g}^A \mapsto \op^A(f_L) \, \hat{g}^A \, \op^A(f_R) 
	, 
\end{align*}
where $\hat{g}^A \in \bounded \bigl ( L^2(\R^d) \bigr )$ is any bounded operator on $L^2(\R^d)$ and $\op^A(f_{L,R})$ are the magnetic Weyl quantizations of symbols $f_{L,R}$. To this end we introduce the magnetic super Weyl system
\begin{align}
	W^A(\Xbf) \, \hat{g}^A = w^A(X_L) \, \hat{g}^A \, w^A(X_R) 
	.
	\label{setting:eqn:magnetic_super_Weyl_system}
\end{align}
Here we have denoted a point in the doubled phase space $\Pspace := \pspace\times\pspace$ by $\Xbf = (X_L,X_R)$, $X_L,X_R \in \pspace$.

We also introduce the symplectic Fourier transform on $\Pspace$, 
\begin{align*}
	(\Fourier_{\Sigma} F)(\Xbf) = \frac{1}{(2\pi)^{2d}} \int_{\Pspace} \dd \Ybf \, \e^{+\ii\Sigma(\Xbf,\Ybf)} \, F(\Ybf) 
	,
\end{align*}
with respect to the symplectic form $\Sigma(\Xbf,\Ybf) := \sigma(X_L,Y_L) + \sigma(X_R,Y_R)$, $\Xbf , \Ybf \in \Pspace$, on doubled phase space $\Pspace$. As before, $\Fourier_{\Sigma}^2 = \id$ holds. 

The symplectic Fourier transform $\Fourier_{\Sigma}$ and the magnetic super Weyl system~\eqref{setting:eqn:magnetic_super_Weyl_system} enable us to introduce the notion of magnetic \emph{super} Weyl quantization
\begin{align}
	\Op^A(F) \, \hat{g}^A := \frac{1}{(2\pi)^{2d}} \int_{\Pspace} \dd\Xbf \, (\Fourier_\Sigma F)(\Xbf) \, W^A(\Xbf) \, \hat{g}^A 
	,
	\label{setting:eqn:super_Weyl_quantization}
\end{align}
whose definition is completely analogous to \eqref{setting:eqn:magnetic_Weyl_quantization}. For a symbol of the form $F(\Xbf) = f_L(X_L) \, f_R(X_R)$ this prescription yields the super operator
\begin{align*}
	\Op^A(F) \, \hat{g}^A = \op^A(f_L) \, \hat{g}^A \, \op^A(f_R) 
	, 
	&&
	\hat{g}^A \in \bounded \bigl ( L^2(\R^d) \bigr ) 
	.
\end{align*}
With abuse of notation we shall denote this super operator by $\op^A(f_L) \otimes \op^A(f_R) = \Op^A \bigl ( f_L \otimes f_R \bigr )$.

One of the fundamental results concerning the magnetic Weyl pseudodifferential super operators is that they map pseudodifferential operators to pseudodifferential operators. More precisely, if $F \in \Schwartz(\Pspace)$ and $g \in \Schwartz(\pspace)$, there is a unique Schwartz function $F \semisuper^B g \in \Schwartz(\pspace)$ such that
\begin{align*}
	\op^A \bigl ( F \semisuper^B g \bigr ) := \Op^A(F) \, \op^A(g) .
\end{align*}
The function $F\semisuper^B g$ is called the magnetic Weyl semi-super product of $F$ and $g$, and the explicit formula for $F\semisuper^B g$ can be found in~\cite[Section~III]{Lee_Lein:magnetic_pseudodifferential_super_operators_basics:2022}. Furthermore, the map $(F,g) \mapsto F\semisuper^B g$ gives rise to a continuous bilinear map from $\Schwartz(\Pspace)\times\Schwartz(\pspace)$ to $\Schwartz(\pspace)$ (see~\cite[Proposition~IV.9]{Lee_Lein:magnetic_pseudodifferential_super_operators_basics:2022}).

Given $F,G \in \Schwartz(\Pspace)$, the magnetic super Weyl product $F \super^B G$ is the Schwartz class function satisfying
\begin{align*}
	\Op^A \bigl ( F \super^B G \bigr ) := \Op^A(F) \, \Op^A(G) .
\end{align*}
The super Weyl product $(F,G) \mapsto F \super^B G$ yields a continuous bilinear map from $\Schwartz(\Pspace) \times \Schwartz(\Pspace)$ to $\Schwartz(\Pspace)$ (\cf \cite[Proposition~IV.14]{Lee_Lein:magnetic_pseudodifferential_super_operators_basics:2022}).

All the results concerning the calculus of magnetic Weyl pseudodifferential super operators described above can be extended to Hörmander symbol classes. In addition to the usual classes of Hörmander symbols, we introduce another class of Hörmander symbols that allow one to study the behavior in left and right momenta separately.
\begin{definition}[Double Hörmander symbol classes]\label{setting:defn:double_Hoermander_symbol_classes}
	Let $0\leq\delta\leq\rho\leq 1$, $\delta<1$ and $m_L,m_R \in \R$. The space $S_{\rho,\delta}^{m_L,m_R}(\Pspace)$ consists of smooth functions $F : \Pspace \longrightarrow \C$ such that, for all $a_L,a_R,\alpha_L,\alpha_R\in\N_0^d$ there exists $C_{a_La_R\alpha_L\alpha_R}>0$ such that
	\begin{align} 
		\babs{\partial_{x_L}^{a_L} \partial_{\xi_L}^{\alpha_L} \partial_{x_R}^{a_R} \partial_{\xi_R}^{\alpha_R} F(\Xbf)} \leq C_{a_L a_R \alpha_L \alpha_R} \, \sexpval{\xi_L}^{m_L - \abs{\alpha_L} \rho + \abs{a_L} \delta} \, \sexpval{\xi_R}^{m_R - \abs{\alpha_R} \rho + \abs{a_R} \delta} 
		&&
		\forall \Xbf\in\Pspace 
		.
		\label{setting:eqn:Hoermander_symbol_separate_estimates}
	\end{align}
\end{definition}
We endow $S_{\rho,\delta}^{m_L,m_R}(\Pspace)$ with the Fréchet space topology generated by the seminorms
\begin{align*}
	\norm{F}_{m_Lm_R,a_La_R\alpha_L\alpha_R} := \sup_{\Xbf\in\Pspace} \Big( \sexpval{\xi_L}^{-m_L+\abs{\alpha_L}\rho-\abs{a_L}\delta} \, \sexpval{\xi_R}^{-m_R+\abs{\alpha_R}\rho-\abs{a_R}\delta} \big| \partial_{x_L}^{a_L} \partial_{\xi_L}^{\alpha_L} \partial_{x_R}^{a_R} \partial_{\xi_R}^{\alpha_R} F(\Xbf) \big| \Big) .
\end{align*}
Unfortunately, there is no simple nesting relation between the Hörmander symbol classes $S_{\rho,\delta}^{m_L,m_R}(\Pspace)$ and $S_{\rho,\delta}^m(\Pspace)$, $m_L,m_R,m \in \R$. The only exception is the inclusion 
\begin{align}
	S^{m_L,m_R}_{\rho,0}(\Pspace) \subseteq S^{\sabs{m_L} + \sabs{m_R}}_{0,0}(\Pspace) 
	\label{setting:eqn:nesting_double_Hoermander_class_Hoermander_class}
\end{align}
that we will use in the proof of Corollary~\ref{boundedness_super_PsiDOs:cor:boundedness_super_PsiDOs}. We refer to \cite[Remark~VI.3]{Lee_Lein:magnetic_pseudodifferential_super_operators_basics:2022} for more details on the relation between two Hörmander symbol classes.

We recall the following results on the semi-super product and super Weyl product of Hörmander symbols:
\begin{proposition}[{{{\cite[Proposition~VI.4]{Lee_Lein:magnetic_pseudodifferential_super_operators_basics:2022}}}}]
	Let $0\leq\rho\leq 1$. Then the following holds.
	\begin{enumerate}[(1)]
		\item Let $m,m' \in \R$. Then the map $(F,g)\mapsto F\semisuper^B g$ gives rise to a continuous bilinear map
		\begin{align*}
			\semisuper^B : S_{\rho,0}^m(\Pspace) \times S_{\rho,0}^{m'}(\pspace) \longrightarrow S_{\rho,0}^{m+m'}(\pspace) .
		\end{align*}
		\item Let $m,m_L,m_R \in \R$. Then the map $(F,g)\mapsto F\semisuper^B g$ gives rise to a continuous bilinear map
		\begin{align*}
			\semisuper^B : S_{\rho,0}^{m_L,m_R}(\Pspace) \times S_{\rho,0}^m(\pspace) \longrightarrow S_{\rho,0}^{m_L+m_R+m}(\pspace) .
		\end{align*}
	\end{enumerate}
\end{proposition}
\begin{proposition}[{{{\cite[Proposition~VI.5]{Lee_Lein:magnetic_pseudodifferential_super_operators_basics:2022}}}}]
	Let $0\leq\rho\leq 1$. Then the following holds.
	\begin{enumerate}[(1)]
		\item Let $m,m' \in \R$. Then the map $(F,G)\mapsto F\super^B G$ gives rise to a continuous bilinear map
		\begin{align*}
			\super^B : S_{\rho,0}^m(\Pspace) \times S_{\rho,0}^{m'}(\Pspace) \longrightarrow S_{\rho,0}^{m+m'}(\Pspace) .
		\end{align*}
		\item Let $m_L,m_L',m_R,m_R' \in \R$. Then the map $(F,G)\mapsto F\super^B G$ gives rise to a continuous bilinear map
		\begin{align*}
			\super^B : S_{\rho,0}^{m_L,m_R}(\Pspace) \times S_{\rho,0}^{m_L',m_R'}(\Pspace) \longrightarrow S_{\rho,0}^{m_L+m_L',m_R+m_R'}(\Pspace) .
		\end{align*}
	\end{enumerate}
\end{proposition}
We can also derive the following result from \cite{Lee_Lein:magnetic_pseudodifferential_super_operators_basics:2022}:
\begin{proposition}
	\label{setting:prop:Hoermander_double_symbol_Schwartz_symbol_semi_super_product}
	Let $0 \leq \rho \leq 1$. Then the following holds:
	\begin{enumerate}[(1)]
		\item Let $F\in S_{\rho,0}^m(\Pspace)$, $m \in \R$ and $g \in \Schwartz(\pspace)$. Then $F \semisuper^B g \in \Schwartz(\pspace)$.
		\item Let $F\in S_{\rho,0}^{m_L,m_R}(\Pspace)$, $m_L , m_R \in \R$ and $g \in \Schwartz(\pspace)$. Then $F \semisuper^B g \in \Schwartz(\pspace)$.
	\end{enumerate}
\end{proposition}
\begin{proof}
	Suppose that $F \in S_{\rho,0}^m(\Pspace)$ and $g \in \Schwartz(\pspace)$. Then also the transpose $F^{\mathrm{t}}(X_L,X_R) := F(X_R,X_L)$, $X_L,X_R \in \pspace$, of the symbol, obtained by swapping the arguments, lies in the same Hörmander class $S_{\rho,0}^m(\Pspace) \ni F^{\mathrm{t}}$. By \cite[Lemma~V.5]{Lee_Lein:magnetic_pseudodifferential_super_operators_basics:2022} we know that there is an inclusion from $S_{\rho,0}^m(\Pspace)$ into the magnetic semi-super Moyal space $\ssMoyalSpace$ defined by
	\begin{align*}
		\ssMoyalSpace := \Bigl \{ G \in \Schwartz'(\Pspace) \; \; \big \vert \; \; \text{$\Schwartz(\pspace)\ni h\mapsto G^{\mathrm{t}} \semisuper^B h \in \Schwartz(\pspace)$ is continuous and linear} \Bigr \} 
		, 
	\end{align*}
	where we have extended the transpose to tempered distributions by duality (\cf the second displayed equation in \cite[Section~V.B]{Lee_Lein:magnetic_pseudodifferential_super_operators_basics:2022}). Therefore, since $F^{\mathrm{t}} \in S_{\rho,0}^m(\Pspace) \subset \ssMoyalSpace$ and $g \in \Schwartz(\pspace)$ we can deduce that $(F^{\mathrm{t}})^{\mathrm{t}} \semisuper^B g = F \semisuper^B g \in \Schwartz(\pspace)$. This proves the first assertion.
	
	Likewise, if $F \in S_{\rho,0}^{m_L,m_R}(\Pspace)$ is a double Hörmander symbol, then we can combine $F^{\mathrm{t}} \in S_{\rho,0}^{m_R,m_L}(\Pspace)$ with \cite[Lemma~V.5]{Lee_Lein:magnetic_pseudodifferential_super_operators_basics:2022}, $S_{\rho,0}^{m_R,m_L}(\Pspace) \subset \ssMoyalSpace$. Hence, the second assertion can be proved along similar lines as that of the proof of the first one. This completes the proof.
\end{proof}
Another notable result in~\cite{Lee_Lein:magnetic_pseudodifferential_super_operators_basics:2022} is the boundedness of magnetic pseudodifferential super operators associated with Schwartz class symbols.
\begin{proposition}[{{{\cite[Lemma~IV.6]{Lee_Lein:magnetic_pseudodifferential_super_operators_basics:2022}}}}]
	Let $F \in \Schwartz(\Pspace)$. Then the following holds.
	\begin{enumerate}[(1)]
		\item $\Op^A(F)$ gives rise to a bounded linear operator from $\bounded \big( L^2(\R^d) \big)$ to itself.
		\item For every $p\geq 1$, $\Op^A(F)$ gives rise to a bounded linear operator from $\mathfrak{L}^p \big( \bounded \big( L^2(\R^d) \big) \big)$ to itself.
	\end{enumerate}
\end{proposition}
As mentioned in the introduction the main goal of this article is to extend the boundedness result on $\mathfrak{L}^2 \bigl ( \bounded \bigl ( L^2(\R^d) \bigr ) \bigr )$ to a suitable class of Hörmander symbols. We do this by representing linear (super) operators as infinite matrices via tight Gabor frames described in \cite{Cornean_Helffer_Purice:boundedness_magnetic_PsiDOs_via_Gabor_frames:2022}.
% subsection magnetic_PsiD_super_Os (end)
% section setting_and_fundamental_definitions (end)
%!TEX root = ./humble boundedness super PsiDOs.tex
\section{Abstract setting: characterization of (super) operators via their matrix elements} % (fold)
\label{abstract_matrix_representation}
Before we can characterize super operators $\hat{F}^A$ through their matrix elements, we need to understand what elements of the $p$-Schatten classes $\mathfrak{L}^p \bigl ( \mathcal{B} \bigl ( L^2(\R^d) \bigr ) \bigr )$ look like — at least for Hilbert-Schmidt operators ($p = 2$). We will also explain our convention of matrix indices for super operators.

While the results in this section are all well-known for orthonormal bases, their extension to various types of frames is not. During the preparation of our manuscript, we came across a very nicely written work by Bingyang, Khoi and Zhu \cite{Bingyang_Khoi_Zhu:frames_Schatten_classes:2015}, which characterized operators from $p$-Schatten classes in terms of $\ell^p$ summability. Importantly, they showed one needs to distinguish between the case $0 < p \leq 2$ and $2 \leq p < \infty$. 

To clean up the presentation, we will discuss the matter in more generality:
\begin{assumption}[General setting]\label{abstract_matrix_representation:assumption:general_setting}
	\begin{enumerate}[(a)]
		\item The Hilbert space $\Hil$ is infinite-dimensional and separable.
		\item $\{ \mathcal{G}_{\alpha} \}_{\alpha \in \Gamma}$ is a Parseval frame indexed by some countable set $\Gamma \cong \Z^k$ for some $k \in \N$.
	\end{enumerate}
\end{assumption}
Indeed, this allows us to omit ${}^A$ in \eqref{setting:eqn:definition_Gabor_frame} and replace pairs of multi indices $(\alpha,\alpha^*)$ with a single one. We also point out that the assumption $\Gamma \cong \Z^k$ could be relaxed as any countably infinite set can be mapped bijectively onto $\N$. Our assumption that $\Hil$ is infinite-dimensional excludes the uninteresting case $\Hil \cong \C^n$ where all $p$-Schatten classes
\begin{align*}
	\mathfrak{L}^p \bigl ( \bounded(\Hil) \bigr ) = \bounded(\Hil) \cong \mathrm{Mat}_{\C}(n) 
	, 
	&&
	1 \leq p < \infty 
	, 
\end{align*}
coincide with the bounded operators (essentially $n \times n$ matrices). What matters is that these assumptions are satisfied for case at hand, namely $\Hil = L^2(\R^d)$ and the Parseval frame from Section~\ref{setting}.

\subsection{Matrix representation of operators} % (fold)
\label{abstract_matrix_representation:operators}
The idea is to write any suitable operator 
\begin{align}
	\hat{f} &= \sum_{\alpha , \beta \in \Gamma} \hat{f}_{\alpha,\beta} \, \sopro{\mathcal{G}_{\alpha}}{\mathcal{G}_{\beta}} 
	\in \mathfrak{L}^p \bigl ( \mathcal{B}(\Hil) \bigr ) 
	, 
	&&
	1 \leq p < \infty 
	, 
	\label{abstract_matrix_representation:eqn:matrix_representation_operator}
\end{align}
as a sum of rank-$1$ operators, where the sum on the right-hand side has to converge in the relevant Banach space. Defining the matrix coefficients as 
\begin{align}
	\hat{f}_{\alpha,\beta} := \bscpro{\mathcal{G}_{\alpha}}{\hat{f} \, \mathcal{G}_{\beta}} 
	\label{abstract_matrix_representation:eqn:matrix_element_operator}
\end{align}
allows us to write operators as the sum over rank-$1$ operators, 
\begin{align*}
	\hat{f} \varphi &= \sum_{\alpha \in \Gamma} \bscpro{\mathcal{G}_{\alpha}}{\hat{f} \varphi} \, \mathcal{G}_{\alpha} 
	% \\
	% &
	= \sum_{\alpha , \beta \in \Gamma} \bscpro{\mathcal{G}_{\alpha}}{\hat{f} \mathcal{G}_{\beta}} \, \sscpro{\mathcal{G}_{\beta}}{\varphi} \, \mathcal{G}_{\alpha} 
	\\
	&= \sum_{\alpha , \beta \in \Gamma} \hat{f}_{\alpha,\beta} \, \sopro{\mathcal{G}_{\alpha}}{\mathcal{G}_{\beta}} \varphi 
	. 
\end{align*}

%%%%%%%%%%%% seems not necessary

%Using the Parseval property, we conclude that the square of the norms in both representations are related by 
%
%\begin{align*}
%	\bnorm{\hat{f} \varphi}^2 &= \sum_{\alpha \in \Gamma} \babs{\bscpro{\mathcal{G}_{\alpha}}{\hat{f} \varphi}}^2 
%	\\
%	&= \sum_{\alpha \in \Gamma} \abs{\sum_{\beta \in \Gamma} \hat{f}_{\alpha,\beta} \, \sscpro{\mathcal{G}_{\beta}}{\varphi}}^2 
%	= \sum_{\alpha \in \Gamma} \abs{\sum_{\beta \in \Gamma} \hat{f}_{\alpha,\beta} \, \varphi_{\beta}}^2 
%	. 
%\end{align*}
%
%Note that we may \emph{not} pull out the sum over $\beta$ and write 
%
%\begin{align*}
%	\sum_{\alpha , \beta \in \Gamma} \babs{\hat{f}_{\alpha,\beta}}^2 \, \babs{\varphi_{\beta}}^2 \neq \sum_{\alpha \in \Gamma} \abs{\sum_{\beta \in \Gamma} \hat{f}_{\alpha,\beta} \, \varphi_{\beta}}^2 
%	= \bnorm{\hat{f} \varphi}^2 
%	. 
%\end{align*}
%

\subsubsection{Hilbert-Schmidt operators} % (fold)
\label{abstract_matrix_representation:operators:hilbert_schmidt_operators}
In this subsection we will characterize the matrix elements of Hilbert-Schmidt operators. They are a special case, because we can make use of the Hilbert space structure and exploit various Hilbert space isometries. 
% FIXME Is that really true? —> Section 3.1.3 
In principle, our result is covered by combining Theorems~A and B in \cite{Bingyang_Khoi_Zhu:frames_Schatten_classes:2015}. 
However, we feel the proof is so short and elegant that we give here it nonetheless. Moreover, our proof has the advantage of giving us the equality of norms in a straightforward fashion. 
\begin{proposition}\label{abstract_matrix_representation:prop:Hilbert_Schmidt}
	\begin{enumerate}[(1)]
		\item There exists an isometry that maps the space of Hilbert-Schmidt operators 
		\begin{align*}
			\mathfrak{U}_2 : \mathfrak{L}^2 \bigl ( \mathcal{B}(\Hil) \bigr ) \longrightarrow \ell^2(\Gamma^2)
		\end{align*}
		into the space of infinite matrix elements whose entries are square summable. This is accomplished via the isometry 
		\begin{align*}
			\mathfrak{U}_2 : \hat{f} \mapsto \bigl ( \hat{f}_{\alpha,\beta} \bigr )_{\alpha , \beta \in \Gamma}
			. 
		\end{align*}
		\item Since $\mathfrak{U}_2$ is an isometry, we can compute the Hilbert-Schmidt norm from the $\ell^2$ norm of its matrix elements, 
		\begin{align*}
			\bnorm{\hat{f}}_{\mathcal{L}^2(\mathcal{B}(\Hil))} &= \bnorm{\mathfrak{U}_2 \hat{f}}_{\ell^2(\Gamma^2)}
			= \bnorm{\bigl ( \hat{f}_{\alpha,\beta} \bigr )_{\alpha , \beta \in \Gamma}}_{\ell^2(\Gamma^2)}
			. 
		\end{align*}
		\item When $\{ \mathcal{G}_{\alpha} \}_{\alpha \in \Gamma}$ is in addition an orthonormal basis, then $\mathfrak{U}_2$ is a Hilbert space isomorphism. 
	\end{enumerate}
\end{proposition}
\begin{proof}
	First of all, the Hilbert space of Hilbert-Schmidt operators 
	\begin{align*}
		\mathfrak{L}^2 \bigl ( \mathcal{B}(\Hil) \bigr ) \cong \Hil \otimes \Hil 
	\end{align*}
	acting on a Hilbert space $\Hil$ can be identified with the (Hilbert space) tensor product of $\Hil$ with itself; this follows from \cite[Theorem~VI.23]{Reed_Simon:M_cap_Phi_1:1972} and using the assumed separability of $\Hil$. In fact, this identification is an isometry, and the Hilbert-Schmidt norm of $\hat{f}$ equals the $\Hil \otimes \Hil$ norm of its image. 
	
	As tensor products of Parseval frames are again Parseval frames (\cf \cite[Theorem~2.3]{Khosravi_Asgari:tensor_product_normalized_tight_frame:2003}), the map 
	\begin{align*}
		\hat{f} \mapsto \bigl ( \hat{f}_{\alpha,\beta} \bigr )_{\alpha , \beta \in \Gamma} 
	\end{align*}
	onto the infinite matrix is also an isometry. Therefore, the composition $\mathfrak{U}_2$ of these two isometries is also an isometry; hence, we have also obtained the claimed equality of norms. 
	
	Lastly, concerning (3), when the Parseval frame is also an orthonormal basis, $\mathfrak{U}_2$ is onto as any element of $\Hil \otimes \Hil$ defines a Hilbert-Schmidt operator and does so uniquely. 
\end{proof}
\begin{corollary}\label{abstract_matrix_representation:cor:Hilbert_Schmidt}
	% CHANGED Isometric isomorphism variant 
	%
	\begin{enumerate}[(1)]
		\item We can interpret the map 
		\begin{align*}
			\mathfrak{U}_2 : \mathfrak{L}^2 \bigl ( \mathcal{B}(\Hil) \bigr ) \longrightarrow \ell^2_{\Hil}(\Gamma^2)
		\end{align*}
		as an isometric Banach space isomorphism onto 
		\begin{align*}
			\ell^2_{\Hil}(\Gamma^2) := \ell^2_{\Hil}(\Gamma) \otimes \ell^2_{\Hil}(\Gamma) 
			. 
		\end{align*}
		\item If in addition $\{ \mathcal{G}_{\alpha} \}_{\alpha \in \Gamma}$ is an orthonormal basis, then $\ell^2_{\Hil}(\Gamma^2) = \ell^2(\Gamma^2)$ and $\mathfrak{U}_2$ is a unitary. 
	\end{enumerate}
\end{corollary}
%
% subsubsection hilbert_schmidt_operators (end)
% subsection matrix_representation_of_operators (end)

\subsection{Matrix representation of super operators} % (fold)
\label{abstract_matrix_representation:matrix_representation_of_super_operators}
The definition of matrix elements of super operators $\hat{F}$ acting on $\mathfrak{L}^2 \bigl ( \mathcal{B}(\Hil) \bigr )$ is conceptually trivial, but the relation between the matrix representation and the matrix elements is not as easy as that between \eqref{abstract_matrix_representation:eqn:matrix_representation_operator} and \eqref{abstract_matrix_representation:eqn:matrix_element_operator}. 

Our convention is such that also here the super operator is (at least formally) the sum 
\begin{align}
	\hat{F} &= \sum_{\alpha_L , \beta_L , \alpha_R , \beta_R \in \Gamma} \hat{F}_{\alpha_L , \beta_L , \alpha_R , \beta_R} \, \sopro{\mathcal{G}_{\alpha_L}}{\mathcal{G}_{\beta_L}} \otimes  \sopro{\mathcal{G}_{\alpha_R}}{\mathcal{G}_{\beta_R}} 
	, 
	\label{abstract_matrix_representation:eqn:matrix_representation_super_operator}
\end{align}
where product operators $\hat{F} = \hat{f}_L \otimes \hat{f}_R$, $\hat{f}_{L,R} \in \mathcal{B}(\Hil)$, by definition act as 
\begin{align*}
	\hat{f}_L \otimes \hat{f}_R (\hat{g}) := \hat{f}_L \, \hat{g} \, \hat{f}_R
	. 
\end{align*}
Our use of $\otimes$ constitutes an abuse of notation, it does \emph{not} denote the tensor product of two operators, which would act on $\Hil \otimes \Hil$. 

It turns out the correct definition of matrix elements is 
\begin{align}
	\hat{F}_{\alpha_L , \beta_L , \alpha_R , \beta_R} := \trace_{\Hil} \Bigl ( \sopro{\mathcal{G}_{\beta_R}}{\mathcal{G}_{\alpha_L}} \, \hat{F} \, \bigl ( \sopro{\mathcal{G}_{\beta_L}}{\mathcal{G}_{\alpha_R}} \bigr ) \Bigr ) 
	. 
	\label{abstract_matrix_representation:eqn:matrix_element_super_operator}
\end{align}
As long as $\hat{F} \, \sopro{\mathcal{G}_{\beta_L}}{\mathcal{G}_{\alpha_R}}$ yields at least a \emph{bounded} operator, the super operator's matrix element $\hat{F}_{\alpha_L , \beta_L , \alpha_R , \beta_R}$ is well-defined. 

Expression~\eqref{abstract_matrix_representation:eqn:matrix_element_super_operator} can be recast as a scalar product, 
\begin{align}
	\hat{F}_{\alpha_L , \beta_L , \alpha_R , \beta_R} &= \bscpro{\mathcal{G}_{\alpha_L}}{\hat{F} \, \bigl ( \sopro{\mathcal{G}_{\beta_L}}{\mathcal{G}_{\alpha_R}} \bigr ) \, \mathcal{G}_{\beta_R}}
	, 
	\label{abstract_matrix_representation:eqn:matrix_element_super_operator_scalar_product}
\end{align}
which is advantageous for some of the computations. We can verify all of these equations for product super operators $\hat{F} = \hat{f}_L \otimes \hat{f}_R$: 
\begin{align*}
	\hat{F}_{\alpha_L , \beta_L , \alpha_R , \beta_R} &= \trace \Bigl ( \sopro{\mathcal{G}_{\beta_R}}{\mathcal{G}_{\alpha_L}} \, \hat{f}_L \, \sopro{\mathcal{G}_{\beta_L}}{\mathcal{G}_{\alpha_R}} \hat{f}_R \Bigr )
	\\
	&= \bscpro{\mathcal{G}_{\alpha_L}}{\hat{f}_L \, \sopro{\mathcal{G}_{\beta_L}}{\mathcal{G}_{\alpha_R}} \, \hat{f}_R \, \mathcal{G}_{\beta_R}}
	\\
	&= \bscpro{\mathcal{G}_{\alpha_L}}{\hat{f}_L \, \mathcal{G}_{\beta_L}} \, \bscpro{\mathcal{G}_{\alpha_R}}{\hat{f}_R \, \mathcal{G}_{\beta_R}}
	= \hat{f}_{L,\alpha_L,\beta_L} \, \hat{f}_{R,\alpha_R,\beta_R}
\end{align*}

\subsection{Expressing products of super operators and operators in terms of their matrix elements} % (fold)
\label{abstract_matrix_representation:operator_products_matrix_products}
When expressing products of (super) operators in terms of their matrix elements, we can \emph{in spirit} just think of those operators as matrix products. The case 
\begin{align*}
	\hat{f} \, \hat{g} \mapsto \left ( \sum_{\gamma \in \Gamma} \hat{f}_{\alpha , \gamma} \, \hat{g}_{\gamma , \beta} \right )_{\alpha , \beta \in \Gamma} 
\end{align*}
was already covered in \cite[Proposition~3.4]{Cornean_Helffer_Purice:boundedness_magnetic_PsiDOs_via_Gabor_frames:2022}. Consequently, we will only deal with the two remaining cases. Here, we need to pay more attention what indices need to be summed over.
% FIXME The sums in the lemma below are not well-defined … at least we do not specify conditions on the matrix coefficients that ensure convergence. 
% Solution: (1) State lemma for p = 2 , \infty
% (2) Add a Remark below the Lemma
%
\begin{lemma}\label{abstract_matrix_representation:lem:product_super_opertor_operator_in_terms_of_matrix_elements}
	Suppose we are in the setting of Assumption~\ref{abstract_matrix_representation:assumption:general_setting}. Let 
	$\hat{F}$ and $\hat{G}$ be elements of $\mathcal{B} \bigl ( \mathfrak{L}^2 \bigl ( \mathcal{B}(\Hil) \bigr ) \bigr )$ and $\hat{g} \in \mathfrak{L}^2 \bigl ( \mathcal{B}(\Hil) \bigr )$. 
	\begin{enumerate}[(1)]
		\item Then we can express $\hat{F} \hat{g}$ in terms of the matrix elements as 
		\begin{align}
			\bigl ( \hat{F} \, \hat{g} \bigr )_{\alpha,\beta} = \sum_{\alpha' , \beta' \in \Gamma} \hat{F}_{\alpha,\beta',\alpha',\beta} \, \hat{g}_{\beta' , \alpha'} 
			. 
			\label{abstract_matrix_representation:eqn:product_super_opertor_operator_in_terms_of_matrix_elements}
		\end{align}
		\item Likewise, the product $\hat{F} \, \hat{G} \in \mathcal{B} \bigl ( \mathfrak{L}^2 \bigl ( \mathcal{B}(\Hil) \bigr ) \bigr )$ has the matrix representation 
		\begin{align}
			\bigl ( \hat{F} \, \hat{G} \bigr )_{\alpha_L,\beta_L,\alpha_R,\beta_R} &= \sum_{\alpha , \beta \in \Gamma} \hat{F}_{\alpha_L,\alpha,\beta,\beta_R} \, \hat{G}_{\alpha,\beta_L,\alpha_R,\beta}
			. 
			\label{abstract_matrix_representation:eqn:product_super_opertor_super_operator_in_terms_of_matrix_elements}
		\end{align}
	\end{enumerate}
\end{lemma}
We will give a proof momentarily. A faster, but formal argument goes as follows: let us assume for a moment that $\{ \mathcal{G}_{\alpha} \}_{\alpha \in \Gamma}$ is not just a Parseval frame, but an orthonormal basis. Then we can obtain an equation for $\bigl ( \hat{F} \, \hat{g} \bigr )_{\alpha,\beta}$ by considering 
\begin{align*}
	\sopro{\mathcal{G}_{\alpha_L}}{\mathcal{G}_{\beta_L}} \otimes  \sopro{\mathcal{G}_{\alpha_R}}{\mathcal{G}_{\beta_R}} \; \bigl ( \sopro{\mathcal{G}_{\alpha}}{\mathcal{G}_{\beta}} \bigr ) &= \sopro{\mathcal{G}_{\alpha_L}}{\mathcal{G}_{\beta_L}} \, \sopro{\mathcal{G}_{\alpha}}{\mathcal{G}_{\beta}} \, \sopro{\mathcal{G}_{\alpha_R}}{\mathcal{G}_{\beta_R}} 
	\\
	&= \scpro{\mathcal{G}_{\beta_L}}{\mathcal{G}_{\alpha}} \, \scpro{\mathcal{G}_{\beta}}{\mathcal{G}_{\alpha_R}} \, \sopro{\mathcal{G}_{\alpha_L}}{\mathcal{G}_{\beta_R}} 
	\\
	&= \delta_{\beta_L,\alpha} \, \delta_{\beta,\alpha_R} \, \sopro{\mathcal{G}_{\alpha_L}}{\mathcal{G}_{\beta_R}} 
	. 
\end{align*}
Consequently, we arrive at equation~\eqref{abstract_matrix_representation:eqn:product_super_opertor_operator_in_terms_of_matrix_elements}. 
\begin{proof}
	\begin{enumerate}[(1)]
		\item To extend these arguments to Parseval frames, we start from first principles and write out the matrix element 
		\begin{align}
			\bigl ( \hat{F} \, \hat{g} \bigr )_{\alpha,\beta} &= \bscpro{\mathcal{G}_{\alpha}}{\hat{F} \, \hat{g} \, \mathcal{G}_{\beta}}
			\notag \\
			&= \sum_{\alpha',\beta' \in \Gamma} \hat{g}_{\beta',\alpha'} \; \bscpro{\mathcal{G}_{\alpha}}{\hat{F} \, \bigl ( \sopro{\mathcal{G}_{\beta'}}{\mathcal{G}_{\alpha'}} \bigr ) \, \mathcal{G}_{\beta}}
			. 
			\label{abstract_matrix_representation:eqn:product_super_opertor_operator_in_terms_of_matrix_elements_trace}
		\end{align}
		We recognize that the scalar product gives $\hat{F}_{\alpha,\beta',\alpha',\beta}$ (\cf equation~\eqref{abstract_matrix_representation:eqn:matrix_element_super_operator_scalar_product}). The sums converge by assumption on $\hat{F}$ and $\hat{g}$. 
		\item To simplify our arguments, we will first proceed under the assumption that $\{ \mathcal{G}_{\alpha} \}_{\alpha \in \Gamma}$ is an orthonormal basis. Applying the tensor product of rank-$1$ operators to some suitable bounded operator $\hat{h}$ yields 
		\begin{align*}
			&\sopro{\mathcal{G}_{\alpha_L}}{\mathcal{G}_{\beta_L}} \otimes  \sopro{\mathcal{G}_{\alpha_R}}{\mathcal{G}_{\beta_R}} \; \Bigl ( \sopro{\mathcal{G}_{\alpha_L'}}{\mathcal{G}_{\beta_L'}} \otimes  \sopro{\mathcal{G}_{\alpha_R'}}{\mathcal{G}_{\beta_R'}} \; (\hat{h}) \Bigr ) 
			\\
			&\qquad 
			= \sopro{\mathcal{G}_{\alpha_L}}{\mathcal{G}_{\beta_L}} \, \sopro{\mathcal{G}_{\alpha_L'}}{\mathcal{G}_{\beta_L'}} \, \hat{h} \, \sopro{\mathcal{G}_{\alpha_R'}}{\mathcal{G}_{\beta_R'}} \, \sopro{\mathcal{G}_{\alpha_R}}{\mathcal{G}_{\beta_R}}
			\\
			&\qquad 
			= \bscpro{\mathcal{G}_{\beta_L}}{\mathcal{G}_{\alpha_L'}} \, \bscpro{\mathcal{G}_{\beta_R'}}{\mathcal{G}_{\alpha_R}} \, \bscpro{\mathcal{G}_{\beta_L'}}{\hat{h} \, \mathcal{G}_{\alpha_R'}} \, \sopro{\mathcal{G}_{\alpha_L}}{\mathcal{G}_{\beta_R}}
			\\
			&\qquad 
			= \delta_{\beta_L,\alpha_L'} \, \delta_{\alpha_R,\beta_R'} \, \sopro{\mathcal{G}_{\alpha_L}}{\mathcal{G}_{\beta_L'}} \otimes  \sopro{\mathcal{G}_{\alpha_R'}}{\mathcal{G}_{\beta_R}} \, (\hat{h}) 
			. 
		\end{align*}
		Hence, we need to contract the index pairs $(\beta_L,\alpha_L')$ and $(\alpha_R,\beta_R')$ in the sum. Relabeling indices appropriately then gives us \eqref{abstract_matrix_representation:eqn:product_super_opertor_super_operator_in_terms_of_matrix_elements}. 
		
		The extension to Parseval frames then follows as in (1): we use the scalar product representation~\eqref{abstract_matrix_representation:eqn:matrix_element_super_operator_scalar_product} of the matrix element for the super operator $\hat{F} \, \hat{G}$ and successively plug in the formal series for $\hat{F}$ and $\hat{G}$. Convergence is assured by assumption in the relevant Banach space of super operators. 
	\end{enumerate}
\end{proof}
%
% subsection expressing_products_of_super_operators_and_operators_in_terms_of_their_matrix_elements (end)
% subsection matrix_representation_of_super_operators (end)
% section  (end)
%!TEX root = ./humble boundedness super PsiDOs.tex
\section{The main technical result: characterization of symbols of super $\Psi$DOs via their matrix elements} % (fold)
\label{characterization_symbols_matrix_elements_super_PsiDOs}
The previous section dealt with the abstract setting, here we will connect it to pseudodifferential theory. For our purposes the abstract discrete group $\Gamma$ has to be replaced with the product $\Gamma \times \Gamma^*$ of the lattice and its dual lattice from Section~\ref{setting}. Hence, each index $\alpha$ from Section~\ref{abstract_matrix_representation} becomes a \emph{pair} of indices $(\alpha,\alpha^*)$. 

The matrix elements of a magnetic pseudodifferential operator $\op^A(f)$ are defined as 
\begin{align}
	\op^A(f)_{\alpha,\alpha^*,\beta,\beta^*} : \negmedspace &= \bscpro{\mathcal{G}_{\alpha,\alpha^*}^A}{\op^A(f) \, \mathcal{G}_{\beta,\beta^*}^A} 
	\notag \\
	&= \frac{1}{(2\pi)^d} \int_{\pspace} \dd X \, (\Fourier_{\sigma} f)(X) \; \bscpro{\mathcal{G}_{\alpha,\alpha^*}^A}{w^A(X) \, \mathcal{G}_{\beta,\beta^*}^A} 
	, 
	\label{characterization_symbols_matrix_elements_super_PsiDOs:eqn:matrix_element_magnetic_PsiDO}
\end{align}
%
% CHANGED Insert equation number for definition of Parseval frame 
where $\mathcal{G}^A_{\alpha,\alpha^*}$ is the Parseval frame from \cite{Cornean_Helffer_Purice:boundedness_magnetic_PsiDOs_via_Gabor_frames:2022} (\cf equation~\eqref{setting:eqn:definition_Gabor_frame}). 

Likewise, magnetic pseudodifferential \emph{super} operators $\Op^A(F)$ have twice the number of indices, 
\begin{align}
	\Op^A(F)_{\alpha_L,\alpha_L^*,\beta_L,\beta_L^*,\alpha_R,\alpha_R^*,\beta_R,\beta_R^*} := \trace_{L^2(\R^d)} \Bigl ( \sopro{\mathcal{G}^A_{\beta_R,\beta_R^*}}{\mathcal{G}^A_{\alpha_L,\alpha_L^*}} \, \hat{F} \, \sopro{\mathcal{G}^A_{\beta_L,\beta_L^*}}{\mathcal{G}^A_{\alpha_R,\alpha_R^*}} \Bigr ) 
	, 
	\label{characterization_symbols_matrix_elements_super_PsiDOs:eqn:super_operator_matrix_element}
\end{align}
one set for the left and one set for the right variables. As $\bigl \{ \mathcal{G}^A_{\alpha,\alpha^*} \bigr \}_{(\alpha,\alpha^*) \in \Gamma \times \Gamma^*}$ is a Parseval frame, despite its inherent overcompleteness, it still acts like an orthonormal basis in that we can express 
\begin{align}
	\Op^A(F) = \sum_{(\alpha_{L,R},\alpha_{L,R}^*) , (\beta_{L,R},\beta_{L,R}^*) \in \Gamma \times \Gamma^*} &\Op^A(F)_{\alpha_L,\alpha_L^*,\beta_L,\beta_L^*,\alpha_R,\alpha_R^*,\beta_R,\beta_R^*} 
	\notag \\
	&\, \cdot 
	\sopro{\mathcal{G}^A_{\alpha_L,\alpha_L^*}}{\mathcal{G}^A_{\beta_L,\beta_L^*}} \otimes  \sopro{\mathcal{G}^A_{\alpha_R,\alpha_R^*}}{\mathcal{G}^A_{\beta_R,\beta_R^*}} 
	\label{characterization_symbols_matrix_elements_super_PsiDOs:eqn:expansion_super_operator}
\end{align}
as an infinite linear combination of product operators. 

The elegance of the approach of Cornean, Helffer and Purice is that oscillatory integral techniques are only needed when characterizing magnetic pseudodifferential operators in terms of their matrix elements (\cf \cite[Theorem~3.1]{Cornean_Helffer_Purice:boundedness_magnetic_PsiDOs_via_Gabor_frames:2022}). Proofs of other results, which would ordinarily require oscillatory integral techniques, now translate to questions of convergence of infinite sums over the indices that enumerate the matrix elements. 

The next two results, Theorem~\ref{characterization_symbols_matrix_elements_super_PsiDOs:thm:characterization_Hoermander_class_super_PsiDOs} and Corollary~\ref{characterization_symbols_matrix_elements_super_PsiDOs:cor:characterization_Hoermander_class_super_PsiDOs}, are analogs of  \cite[Theorem~3.1]{Cornean_Helffer_Purice:boundedness_magnetic_PsiDOs_via_Gabor_frames:2022}. Their proofs are essentially identical to that given by Cornean, Helffer and Purice, many of the expressions just need to be “doubled” to account for the presence of left and right variables. 
\begin{theorem}\label{characterization_symbols_matrix_elements_super_PsiDOs:thm:characterization_Hoermander_class_super_PsiDOs}
	Suppose we are given a tempered distribution $F \in \Schwartz'(\Pspace)$ and some $m_L , m_R \in \R$. Then the following two statements are equivalent: 
	\begin{enumerate}[(a)]
		\item $F \in S^{m_L,m_R}_{0,0}(\Pspace)$ is a Hörmander symbol of order $(m_L,m_R)$ and type $(0,0)$. 
		\item For any $(n_L,n_R,n_L^*,n_R^*) \in \N_0^4$ there exists a constant $C_{n_L,n_R,n_L^*,n_R^*}(F,B) > 0$ so that the matrix elements satisfy the bound 
		\begin{align}
			\sup_{(\alpha_L,\alpha_L^*),(\alpha_R,\alpha_R^*),(\beta_L,\beta_L^*),(\beta_R,\beta_R^*) \in \Gamma \times \Gamma^*} &\sexpval{\alpha_L - \beta_L}^{n_L} \, \sexpval{\alpha_L^* - \beta_L^*}^{n_L^*} \, \sexpval{\alpha_R - \beta_R}^{n_R} \, \sexpval{\alpha_R^* - \beta_R^*}^{n_R^*} \,
			\notag \\
			&\cdot 
			\sexpval{\alpha_L^* + \beta_L^*}^{-m_L} \, \sexpval{\alpha_R^* + \beta_R^*}^{-m_R} \, \Op^A(F)_{\alpha_L,\alpha_L^*,\beta_L,\beta_L^*,\alpha_R,\alpha_R^*,\beta_R,\beta_R^*} 
			\notag \\
			&\qquad \qquad 
			\leq C_{n_L,n_R,n_L^*,n_R^*}(F,B) 
			. 
			\label{characterization_matrix_elements_super_symbols:eqn:characterization_double_Hoermander_symbols}
		\end{align}
	\end{enumerate}
\end{theorem}
\begin{proof}
	The strategy of the proof is identical to \cite{Cornean_Helffer_Purice:boundedness_magnetic_PsiDOs_via_Gabor_frames:2022}, we will just outline the necessary modifications. 
	
	First of all, the equations for $\Op^A(F)_{\alpha_L,\alpha_L^*,\beta_L,\beta_L^*,\alpha_R,\alpha_R^*,\beta_R,\beta_R^*}$ is essentially a “doubling” of the equation for $\op^A(f)_{\alpha,\alpha^*,\beta,\beta^*}$: indeed, for any combination of multi indices, we can rewrite each matrix element~\eqref{characterization_symbols_matrix_elements_super_PsiDOs:eqn:matrix_element_magnetic_PsiDO} in terms of the symplectic Fourier transform of $f$ and the matrix element of the Weyl system. Writing out the action of the Weyl system and some algebra, we recover \cite[equation~(3.2)]{Cornean_Helffer_Purice:boundedness_magnetic_PsiDOs_via_Gabor_frames:2022}. 
	
	The equivalent expression for the magnetic super Weyl quantization is essentially two copies of \cite[equation~(3.2)]{Cornean_Helffer_Purice:boundedness_magnetic_PsiDOs_via_Gabor_frames:2022}, one for the left and one for the right variables. The reason is that $\Op^A(F)$ is the integral of the product operator $w^A(X_L) \otimes w^A(X_R)$. Hence, from equation~\eqref{characterization_symbols_matrix_elements_super_PsiDOs:eqn:super_operator_matrix_element} and 
	% CHANGED Insert equation number for definition of super Weyl quantization
	the definition of the magnetic super Weyl quantization~\eqref{setting:eqn:super_Weyl_quantization}, we deduce 
	\begin{align*}
		&\Op^A(F)_{\alpha_L,\alpha_L^*,\beta_L,\beta_L^*,\alpha_R,\alpha_R^*,\beta_R,\beta_R^*} 
		% \label{characterization_symbols_matrix_elements_super_PsiDOs:eqn:matrix_element_OpA_F_abstract}
		\\
		&\qquad =
		\frac{1}{(2\pi)^{2d}} \int_{\pspace} \dd X_L \int_{\pspace} \dd X_R \, (\Fourier_{\Sigma} F)(X_L,X_R) \; \bscpro{\mathcal{G}_{\alpha_L,\alpha_L^*}^A}{w^A(X_L) \, \mathcal{G}_{\beta_L,\beta_L^*}^A} \; \bscpro{\mathcal{G}_{\alpha_R,\alpha_R^*}^A}{w^A(X_R) \, \mathcal{G}_{\beta_R,\beta_R^*}^A} 
		\notag 
	\end{align*}
	is the double phase space integral of the “expectation values” of left and right Weyl systems and the symplectic Fourier transform of $F$. 
	
	We will re-use the notation from the proof of \cite[Theorem~3.1]{Cornean_Helffer_Purice:boundedness_magnetic_PsiDOs_via_Gabor_frames:2022} whenever possible. 
	\medskip
	
	\noindent
	“(a) $\Longrightarrow$ (b):” Suppose $F \in S^{m_L,m_R}_{0,0}(\Pspace)$ is a Hörmander symbol of order $(m_L,m_R)$ and fix the four non-negative integers $n_{L,R} \in \N_0$ and $n_{L,R}^* \in \N_0$. 
	
	After applying the variable transformations from \cite{Cornean_Helffer_Purice:boundedness_magnetic_PsiDOs_via_Gabor_frames:2022} to \eqref{characterization_symbols_matrix_elements_super_PsiDOs:eqn:super_operator_matrix_element} twice, once for the left and once for the right variables, one essentially obtains two copies of \cite[equation~(3.7)]{Cornean_Helffer_Purice:boundedness_magnetic_PsiDOs_via_Gabor_frames:2022}. The only difference is that we have to replace $\Phi(z + \nicefrac{\mu}{2} , \zeta + \nicefrac{\mu^*}{2})$ with 
	\begin{align*}
		F \bigl ( z_L + \tfrac{1}{2} \mu_L , \zeta_L + \tfrac{1}{4 \pi} \mu_L^* , z_R + \tfrac{1}{2} \mu_R , \zeta_R + \tfrac{1}{4 \pi} \mu_R^* \bigr ) 
		. 
	\end{align*}
	Repeating the arguments from the proof of \cite[Theorem~3.1]{Cornean_Helffer_Purice:boundedness_magnetic_PsiDOs_via_Gabor_frames:2022} and using Definition~\ref{setting:defn:double_Hoermander_symbol_classes} for double Hörmander symbols then gives us an estimate of the matrix element of the form 
	\begin{align*}
		&\babs{\Op^A(F)_{\alpha_L,\alpha_L^*,\beta_L,\beta_L^*,\alpha_R,\alpha_R^*,\beta_R,\beta_R^*}}
		\\
		&\qquad 
		\leq
		C \, \sexpval{\nu_L^*}^{-2 N_L} \, \sexpval{\nu_L}^{2 N_L + 2 M_L - 2 K_L} \, \sexpval{\nu_R^*}^{-2 N_R} \, \sexpval{\nu_R}^{2 N_R + 2 M_R - 2 K_R}
		\\
		&\qquad \quad \; \cdot \, 
		\int_{\R^d} \dd \zeta_L \, \sexpval{\zeta_L}^{- 2 M_L} \, \sexpval{\zeta_L + \tfrac{1}{4 \pi} \mu_L^*}^{m_L} \, \int_{\R^d} \dd \zeta_R \, \sexpval{\zeta_R}^{- 2 M_R} \, \sexpval{\zeta_R + \tfrac{1}{4 \pi} \mu_R^*}^{m_R} 
		, 
	\end{align*}
	%
	% FIXME Check later: Cornean, Helffer and Purice have two inconsistent definitions for \nu^*: once as \beta^* - \alpha^* (p. 7) and once as \alpha^* - \beta^* (p. 10). Check which conventeion we used when checking their computations. Also, we need to wait for the revision of their paper to fix our convention. 
	where the indices 
	\begin{align*}
		(\mu_{L,R} , \nu_{L,R}) \, &:= \bigl ( \alpha_{L,R} + \beta_{L,R} , \alpha_{L,R} - \beta_{L,R} \bigr )
		, 
		\\
		(\mu_{L,R}^* , \nu_{L,R}^*) \, &:= \bigl ( \alpha_{L,R}^* + \beta_{L,R}^* , \beta_{L,R}^* - \alpha_{L,R}^* \bigr )
		, 
	\end{align*}
	are the sums and differences of left and right multi indices. 
	
	Choosing the integers $N_{L,R}$, $M_{L,R}$ and $K_{L,R}$ large enough then ensures the existence of the oscillatory integral~\eqref{characterization_symbols_matrix_elements_super_PsiDOs:eqn:super_operator_matrix_element}; the least lower bounds for $N_{L,R}$, $M_{L,R}$ and $K_{L,R}$ are those spelled out in \cite{Cornean_Helffer_Purice:boundedness_magnetic_PsiDOs_via_Gabor_frames:2022} and involve $n_{L,R}$ and $n_{L,R}^*$. It turns out we may choose $N_{L,R}$ and $K_{L,R} - N_{L,R} - M_{L,R}$ as large as we like, which establishes at most polynomial growth in the indices $\mu_{L,R}^* = \alpha_{L,R}^* + \beta_{L,R}^*$ with powers $m_{L,R}$ and rapid decay in $\nu_{L,R}$ and $\nu_{L,R}^*$. This shows the estimate from (b) holds true. 
	\medskip
	
	\noindent
	“(b) $\Longrightarrow$ (a):” Suppose for any four non-negative integers $n_{L,R} \in \N_0$ and $n_{L,R}^* \in \N_0$ there exists a constant $C > 0$ so that the estimate from (b) holds true. The idea of the proof is to reconstruct the symbol of the operator as the Wigner transform of the kernel, written as the infinite sum
	\begin{align*}
		\Op^A(F) := \sum_{(\alpha_{L,R},\alpha_{L,R}^*) , (\beta_{L,R},\beta_{L,R}^*) \in \Gamma \times \Gamma^*} &\Op^A(F)_{\alpha_L,\alpha_L^*,\beta_L,\beta_L^*,\alpha_R,\alpha_R^*,\beta_R,\beta_R^*}
		\\
		&\cdot \, 
		\sopro{\mathcal{G}^A_{\alpha_L,\alpha_L^*}}{\mathcal{G}^A_{\beta_L,\beta_L^*}} \otimes \sopro{\mathcal{G}^A_{\alpha_R,\alpha_R^*}}{\mathcal{G}^A_{\beta_R,\beta_R^*}}
		. 
	\end{align*}
	The main point of the proof is to ensure convergence of the above sum in the appropriate sense. 
	
	Hence, we truncate the sum by fixing some $N \in \N_0$ and defining 
	\begin{align*}
		\Op^A(F_N) := \sum_{\sabs{(\alpha_{L,R},\alpha_{L,R}^*)} , \sabs{(\beta_{L,R},\beta_{L,R}^*)} \leq N} &\Op^A(F)_{\alpha_L,\alpha_L^*,\beta_L,\beta_L^*,\alpha_R,\alpha_R^*,\beta_R,\beta_R^*} 
		\\
		&\cdot \, 
		\sopro{\mathcal{G}^A_{\alpha_L,\alpha_L^*}}{\mathcal{G}^A_{\beta_L,\beta_L^*}} \otimes \sopro{\mathcal{G}^A_{\alpha_R,\alpha_R^*}}{\mathcal{G}^A_{\beta_R,\beta_R^*}}
		. 
	\end{align*}
	%
	% Since this is a finite linear combination of product operators acting from the left and the right, we can directly read off the corresponding operator kernel,
	% %
	% \begin{align*}
	% 	K^A_N(x_L,y_L,x_R,y_R) := \sum_{\sabs{(\alpha_{L,R},\alpha_{L,R}^*)} , \sabs{(\beta_{L,R},\beta_{L,R}^*)} \leq N} &\Op^A(F)_{\alpha_L,\alpha_L^*,\beta_L,\beta_L^*,\alpha_R,\alpha_R^*,\beta_R,\beta_R^*}
	% 	\, \cdot \\
	% 	&\cdot \,
	% 	\mathcal{G}^A_{\alpha_L,\alpha_L^*}(x_L) \, \overline{\mathcal{G}^A_{\beta_L,\beta_L^*}(y_L)} \, \mathcal{G}^A_{\alpha_R,\alpha_R^*}(x_R) \, \overline{\mathcal{G}^A_{\beta_R,\beta_R^*}(y_R)}
	% 	.
	% \end{align*}
	% %
	% All elements of our Parseval frame are Schwartz functions, so $K^A_N \in \Schwartz \bigl ( (\R^d \times \R^d)^2 \bigr )$ is a sequence of Schwartz functions; the sequence $K^A_N$ converges in the distributional sense.
	%
	For product functions $\Op^A \bigl ( f_L \otimes f_R \bigr ) = \op^A(f_L) \otimes \op^A(f_R)$ the magnetic \emph{super} Weyl quantization is the “product” of the two (regular) magnetic Weyl quantizations. Therefore, this relationship carries over to the inverse of the super Weyl quantization, which in this case is two copies of the magnetic Wigner transform applied to left and right variables separately (\cf the displayed equation below equation~(3.9) in \cite{Lee_Lein:magnetic_pseudodifferential_super_operators_basics:2022}). 
	
	The magnetic Wigner transform is a topological vector space isomorphism on $\Schwartz$, and hence, for any $N \in \N$ 
	\begin{align*}
		F_N := \sum_{\sabs{(\alpha_{L,R},\alpha_{L,R}^*)} , \sabs{(\beta_{L,R},\beta_{L,R}^*)} \leq N} &\, \Op^A(F)_{\alpha_L,\alpha_L^*,\beta_L,\beta_L^*,\alpha_R,\alpha_R^*,\beta_R,\beta_R^*} 
		\\
		&\cdot \, (\op^A)^{-1} \bigl ( \sopro{\mathcal{G}^A_{\alpha_L,\alpha_L^*}}{\mathcal{G}^A_{\beta_L,\beta_L^*}} \bigr ) \otimes (\op^A)^{-1} \bigl ( \sopro{\mathcal{G}^A_{\alpha_R,\alpha_R^*}}{\mathcal{G}^A_{\beta_R,\beta_R^*}} \bigr ) 
		\in \Schwartz(\Pspace)
	\end{align*}
	is a Schwartz function. This defines a sequence of Schwartz functions $F_N \rightarrow F \in \Schwartz'(\Pspace)$, which converges in the distributional sense to $F$. In fact, repeating the arguments from \cite{Cornean_Helffer_Purice:boundedness_magnetic_PsiDOs_via_Gabor_frames:2022} twice, once for left and once for right variables yields the claim that 
	\begin{align*}
		F_N \xrightarrow{N \rightarrow \infty} F \in S^{m_L,m_R}_{0,0}(\Pspace) 
	\end{align*}
	converges uniformly on compact subsets to a \emph{Hörmander symbol} of order $(m_L,m_R)$. This finishes the proof. 
\end{proof}
We get an analogous result for standard Hörmander symbols (\cf Definition~\ref{setting:defn:Hoermander_symbol_classes}): 
\begin{corollary}\label{characterization_symbols_matrix_elements_super_PsiDOs:cor:characterization_Hoermander_class_super_PsiDOs}
	Suppose we are given a tempered distribution $F \in \Schwartz'(\Pspace)$ and some $m \in \R$. Then the following two statements are equivalent: 
	\begin{enumerate}[(a)]
		\item $F \in S^m_{0,0}(\Pspace)$ is a Hörmander symbol of order $m$ and type $(0,0)$. 
		\item For any $n , n^* \in \N_0^2$ there exists a constant $C_{n,n^*}(F,B) > 0$ so that the matrix elements satisfy the bound 
		% \item For any $(n_{L,1},n_{R,1},n_{L,2},n_{R,2}) \in \N_0^4$ there exists a constant $C_{n_{L,1},n_{R,1},n_{L,2},n_{R,2}}(F,B)$ so that the matrix elements satisfy
		%
		\begin{align}
			\sup_{(\alpha_L,\alpha_L^*),(\alpha_R,\alpha_R^*),(\beta_L,\beta_L^*),(\beta_R,\beta_R^*) \in \Gamma \times \Gamma^*} &\bexpval{(\alpha_L,\alpha_R) - (\beta_L,\beta_R)}^n \, \bexpval{(\alpha_L^*,\alpha_R^*) - (\beta_L^*,\beta_R^*)}^{n^*}
			\notag \\
			&\cdot 
			\bexpval{(\alpha_L^*,\alpha_R^*) + (\beta_L^*,\beta_R^*)}^{-m} \, \Op^A(F)_{\alpha_L,\alpha_L^*,\beta_L,\beta_L^*,\alpha_R,\alpha_R^*,\beta_R,\beta_R^*} 
			\notag \\
			&\qquad \qquad 
			\leq C_{n,n^*}(F,B) 
			. 
			\label{characterization_matrix_elements_super_symbols:eqn:characterization_Hoermander_symbols}
		\end{align}
	\end{enumerate}
\end{corollary}
\begin{proof}
	The proof is essentially identical to that of Theorem~\ref{characterization_symbols_matrix_elements_super_PsiDOs:thm:characterization_Hoermander_class_super_PsiDOs}, and we content ourselves with outlining the modifications. 
	\medskip
	
	\noindent
	“(a) $\Longrightarrow$ (b):” Compared to the proof of Theorem~\ref{characterization_symbols_matrix_elements_super_PsiDOs:thm:characterization_Hoermander_class_super_PsiDOs}, we will use different $L$ operators in the proof. Namely in the notation of \cite[equation~(3.8)]{Cornean_Helffer_Purice:boundedness_magnetic_PsiDOs_via_Gabor_frames:2022} we apply the three operators 
	\begin{align*}
		L_{(z_L,z_R)} \, &:= \bexpval{(\nu_L^*,\nu_R^*)}^{-2} \, \bigl ( 1 - \Delta_{z_L} - \Delta_{z_R} \bigr ) 
		, 
		\\
		L_{(\zeta_L,\zeta_R)} \, &:= \bexpval{(\nu_L,\nu_R)}^{-2} \, \bigl ( 1 - \Delta_{\zeta_L} - \Delta_{\zeta_R} \bigr ) 
		, 
		\\
		L_{(v_L,v_R)} \, &:= \bexpval{(\zeta_L,\zeta_R)}^{-2} \, \bigl ( 1 - \Delta_{v_L} - \Delta_{v_R} \bigr ) 
		. 
	\end{align*}
	If we introduce notation such as $z := (z_L,z_R)$, $\zeta := (\zeta_L,\zeta_R)$, etc., we symbolically recover the expression below \cite[equation~(3.9)]{Cornean_Helffer_Purice:boundedness_magnetic_PsiDOs_via_Gabor_frames:2022} and can perform all subsequent arguments almost verbatim. This shows the first direction. 
	\medskip
	
	\noindent
	“(b) $\Longrightarrow$ (a):” Similarly to before, we can introduce the joint left/right variables $\nu := (\nu_L,\nu_R)$, $\nu^* := (\nu_L^*,\nu_R^*)$, etc.\ to make the expressions in the proof look almost identical to that in \cite{Cornean_Helffer_Purice:boundedness_magnetic_PsiDOs_via_Gabor_frames:2022}. Because the matrix elements now satisfy \eqref{characterization_matrix_elements_super_symbols:eqn:characterization_Hoermander_symbols} instead of \eqref{characterization_matrix_elements_super_symbols:eqn:characterization_double_Hoermander_symbols}, we recover $\sexpval{\zeta}^m = \bexpval{(\zeta_L,\zeta_R)}^m$ on the right-hand side of our estimate of $\babs{\partial_{z_L}^{a_L} \partial_{z_R}^{a_R} \partial_{\zeta_L}^{\alpha_L} \partial_{\zeta_R}^{\alpha_R} F_N(z_L,\zeta_L,z_R,\zeta_R)}$ with a constant that is uniform in $N$, $(z_L,\zeta_L)$ and $(z_R,\zeta_R)$. Consequently, $F_N \rightarrow F \in S^m_{0,0}(\Pspace)$ converges uniformly on compact subsets to a Hörmander symbol of order $m$ and type $(0,0)$. 
\end{proof}
%
% section the_main_technical_result_characterization_of_symbols_of_super_psi_dos_via_their_matrix_elements (end)
%!TEX root = ./humble boundedness super PsiDOs.tex
\section{Boundedness criteria for magnetic pseudodifferential super operators} % (fold)
\label{boundedness_super_PsiDOs}
We are now in a position to formulate and prove boundedness criteria on the level of matrix elements: we have given characterizations of Hilbert-Schmidt operators in Section~\ref{abstract_matrix_representation} and obtained growth/decay estimates on the matrix elements of magnetic pseudodifferential super operators (Theorem~\ref{characterization_symbols_matrix_elements_super_PsiDOs:thm:characterization_Hoermander_class_super_PsiDOs} and Corollary~\ref{characterization_symbols_matrix_elements_super_PsiDOs:cor:characterization_Hoermander_class_super_PsiDOs}). Combining both then gives very simple proofs of boundedness in the spirit of \cite{Cornean_Helffer_Purice:simple_proof_Beals_criterion_magnetic_PsiDOs:2018,Cornean_Helffer_Purice:boundedness_magnetic_PsiDOs_via_Gabor_frames:2022}.

\subsection{Boundedness of magnetic pseudodifferential superoperators} % (fold)
\label{boundedness_super_PsiDOs:super_PsiDOs}
We will need two more auxiliary statements in the proof of boundedness. The first is a version of Schur's test (\cf \cite[Lemma~6.1.7]{Cordero_Rodino:operators_time_frequency:2020}).
\begin{lemma}[Schur's Test]\label{boundedness_super_PsiDOs:lem:Schur}
	Let $\Gamma$ be a countable set and $(K_{\alpha, \beta})_{\alpha, \beta \in \Gamma}$ be a sequence with values in $\C$ indexed by $\Gamma\times\Gamma$. Suppose there is $C>0$ such that
	\begin{align*}
		\sup_{\beta \in \Gamma} \sum_{\alpha \in \Gamma} \abs{K_{\alpha, \beta}} \leq C \quad \text{and} \quad \sup_{\alpha \in \Gamma} \sum_{\beta \in \Gamma} \abs{K_{\alpha, \beta}} \leq C .
	\end{align*}
	Then the map $K$ defined by
	\begin{align*}
		(Kc)_{\alpha} := \sum_{\beta \in \Gamma} K_{\alpha, \beta} \, c_\beta , \qquad (c_\beta)_{\beta \in \Gamma} \in \ell^2(\Gamma)
	\end{align*}
	gives rise to a continuous linear operator from $\ell^2(\Gamma)$ to itself with $\norm{K}_{\mathcal{B} (\ell^2(\Gamma))} \leq C$.
\end{lemma}
The second is a refinement of Lemma~\ref{abstract_matrix_representation:lem:product_super_opertor_operator_in_terms_of_matrix_elements}~(1) with more explicit control over the convergence of the sums. As mentioned in Subsection~\ref{setting:magnetic_PsiDOs} the magnetic Weyl quantization $\op^A$ yields a unitary map from $L^2(\pspace)$ to $\mathfrak{L}^2 \bigl ( \mathcal{B} \bigl ( L^2(\R^d) \bigr ) \bigr )$. Therefore, it follows that $\op^A \bigl ( \Schwartz(\pspace) \bigr ) := \bigl \{ \op^A(f) \; \vert \; f \in \Schwartz(\pspace) \bigr \}$ is dense in $\mathfrak{L}^2 \bigl ( \mathcal{B} \bigl ( L^2(\R^d) \bigr ) \bigr )$ since $\Schwartz(\pspace)$ is a dense subspace of $L^2(\pspace)$. A linear (super) operator $\hat{F} : \op^A \bigl ( \Schwartz(\pspace) \bigr ) \longrightarrow \mathfrak{L}^2 \bigl ( \mathcal{B} \bigl ( L^2(\R^d) \bigr ) \bigr )$ can be regarded as a densely defined unbounded operator in $\mathfrak{L}^2 \bigl ( \mathcal{B} \bigl ( L^2(\R^d) \bigr ) \bigr )$ with domain $\op^A \bigl (\Schwartz(\pspace) \bigr )$. The scalar product on the space of Hilbert-Schmidt operators gives rise to the adjoint $\hat{F}^* : \mathscr{D}(\hat{F}^*) \longrightarrow \mathfrak{L}^2 \bigl ( \mathcal{B} \bigl ( L^2(\R^d) \bigr ) \bigr )$ of $\hat{F}$ by letting
\begin{align*}
	{\scpro{\hat{F}^* \, \hat{f}}{\hat{g}}}_{\mathfrak{L}^2(\mathcal{B}(L^2(\R^d)))} = \scpro{\hat{f}}{\hat{F} \, \hat{g}}_{\mathfrak{L}^2 \mathcal{B} ( L^2(\R^d) ) )} \qquad \forall \hat{f}\in\mathscr{D}(\hat{F}^*) 
	&&
	\forall \hat{g} \in \mathfrak{L}^2 \bigl ( \mathcal{B} \bigl ( L^2(\R^d) \bigr ) \bigr )
	,
\end{align*}
where as usual the domain of the adjoint is defined as 
\begin{align*}
	\mathscr{D}(\hat{F}^*) := \Bigl \{ \hat{f}\in\mathfrak{L}^2 \bigl ( \mathcal{B} ( L^2(\R^d) ) \bigr ) \; \; \big \vert \; \; &\babs{{\scpro{\hat{f}}{\hat{F} \, \hat{g}}}_{\mathfrak{L}^2 (\mathcal{B} ( L^2(\R^d) ) )}} \leq C(\hat{F},\hat{f}) \, \snorm{\hat{g}}_{\mathfrak{L}^2(\mathcal{B}(L^2(\R^d)))} \,\, 
	\Bigr . 
	\\
	&\Bigl . 
	\forall \hat{g}\in\mathfrak{L}^2 \bigl ( \mathcal{B} \bigl ( L^2(\R^d) \bigr ) \bigr ) \Bigr \} 
	.
\end{align*}
The following result is the super operator analog of \cite[Proposition~2.4]{Cornean_Helffer_Purice:boundedness_magnetic_PsiDOs_via_Gabor_frames:2022}.
%
% CHANGED Change index notation from generic to specific case of magnetic PsiDOs
% CHANGED Clarify relation between g \in S and \hat{g}^A = op^A(g) 
\begin{proposition} \label{boundedness_super_PsiDOs:prop:superoperator_image_expansion}
	Suppose that $\hat{F}^A : \op^A \bigl ( \Schwartz(\pspace) \bigr ) \longrightarrow \mathfrak{L}^2 \bigl ( \mathcal{B} \bigl ( L^2(\R^d) \bigr ) \bigr )$ is a densely defined unbounded super operator as above and $\op^A \bigl ( \Schwartz(\pspace) \bigr ) \subset \mathscr{D}(\hat{F}^*)$. Let $g \in \Schwartz(\pspace)$ and $\hat{g}^A = \op^A(g)$ its magnetic Weyl quantization. Then we can express the operator in terms of the matrix elements 
	\begin{align*}
		\hat{F}^A \, \hat{g}^A = \sum_{(\alpha , \alpha^*) , (\beta , \beta^*) \in \Gamma \times \Gamma^*} \left ( \sum_{(\alpha' , {\alpha'}^*) , (\beta' , {\beta'}^*) \in \Gamma \times \Gamma^*} \hat{F}^A_{\alpha , \alpha^* , \beta'  , {\beta'}^* , \alpha' , {\alpha'}^* , \beta , \beta^*} \, \hat{g}^A_{\beta' , {\beta'}^* , \alpha' , {\alpha'}^*} \right ) \, \sopro{\mathcal{G}^A_{\alpha , \alpha^*}}{\mathcal{G}^A_{\beta , \beta^*}}
		,
	\end{align*}
	where the inner sum with respect to $(\alpha' , {\alpha'}^*) , (\beta' , {\beta'}^*) \in \Gamma \times \Gamma^*$ converges in the sense of $\ell^1 \bigl ((\Gamma \times \Gamma^*)^2 \bigr )$ and the outer sum with respect to $(\alpha , \alpha^*) , (\beta , \beta^*) \in \Gamma \times \Gamma^*$ converges in $\mathfrak{L}^2 \bigl ( \mathcal{B} \bigl ( L^2(\R^d) \bigr ) \bigr )$.
\end{proposition}
%
% 
% CHANGED Change index notation from generic to specific case of magnetic PsiDOs
\begin{proof}
	Using \eqref{abstract_matrix_representation:eqn:matrix_representation_operator} we can expand the application of the super operator $\hat{F}^A$ on the Hilbert-Schmidt operator $\hat{g}^A$, 
	\begin{align}
		\hat{F}^A \, \hat{g}^A &= \sum_{(\alpha , \alpha^*) , (\beta , \beta^*) \in \Gamma \times \Gamma^*} {\bscpro{ \sopro{\mathcal{G}^A_{\alpha , \alpha^*}}{\mathcal{G}^A_{\beta , \beta^*}} \, }{ \, \hat{F}^A \, \hat{g}^A}}_{\mathfrak{L}^2(\mathcal{B}(L^2(\R^d)))} \, \sopro{\mathcal{G}^A_{\alpha , \alpha^*}}{\mathcal{G}^A_{\beta , \beta^*}} 
		\notag 
		\\
		&= \sum_{(\alpha , \alpha^*) , (\beta , \beta^*) \in \Gamma \times \Gamma^*} {\Bscpro{ \hat{F}^{A \, \ast} \bigl ( \sopro{\mathcal{G}^A_{\alpha , \alpha^*} }{\mathcal{G}^A_{\beta , \beta^*}} \bigr ) \, }{ \, \hat{g}^A}}_{\mathfrak{L}^2(\mathcal{B}(L^2(\R^d)))} \, \sopro{\mathcal{G}^A_{\alpha , \alpha^*}}{\mathcal{G}^A_{\beta , \beta^*}} 
		,
	\label{abstract_matrix_representation:eqn:image_of_operator_by_superoperator_expansion}
	\end{align}
	where the sum converges in $\mathfrak{L}^2 \bigl ( \mathcal{B} \bigl ( L^2(\R^d) \bigr ) \bigr )$. Note that by using \eqref{abstract_matrix_representation:eqn:matrix_representation_operator} again we also obtain
	\begin{align*}
		\hat{g}^A
		&= \sum_{(\alpha' , {\alpha'}^*) , (\beta' , {\beta'}^*) \in \Gamma \times \Gamma^*} \scpro{\mathcal{G}^A_{\beta' , {\beta'}^*} \, }{ \, \hat{g}^A \, \mathcal{G}^A_{\alpha' , {\alpha'}^*}} \, 
		\sopro{\mathcal{G}^A_{\beta' , {\beta'}^*}}{\mathcal{G}^A_{\alpha' , {\alpha'}^*}} 
		\\
		&= \sum_{(\alpha' , {\alpha'}^*) , (\beta' , {\beta'}^*) \in \Gamma \times \Gamma^*} \hat{g}^A_{\beta' , {\beta'}^* , \alpha' , {\alpha'}^*} \, \sopro{\mathcal{G}^A_{\beta' , {\beta'}^*}}{\mathcal{G}^A_{\alpha' , {\alpha'}^*}} 
		.
	\end{align*}
	By plugging the expansion of $\hat{g}^A$ back into \eqref{abstract_matrix_representation:eqn:image_of_operator_by_superoperator_expansion} we get
	\begin{align*}
		&{\Bscpro{ \hat{F}^{A \, \ast} \bigl ( \sopro{\mathcal{G}^A_{\alpha , \alpha^*} }{\mathcal{G}^A_{\beta , \beta^*}} \bigr ) \, }{ \, \hat{g}^A}}_{\mathfrak{L}^2(\mathcal{B}(L^2(\R^d)))} 
		\\
		&\qquad 
		= \sum_{(\alpha' , {\alpha'}^*) , (\beta' , {\beta'}^*) \in \Gamma \times \Gamma^*} \hat{g}^A_{\beta' , {\beta'}^* , \alpha' , {\alpha'}^*} \, {\Bscpro{ \hat{F}^{A \, \ast} \bigl ( \sopro{\mathcal{G}^A_{\alpha , \alpha^*} }{\mathcal{G}^A_{\beta , \beta^*}} \bigr ) \, }{ \, \sopro{\mathcal{G}^A_{\beta' , {\beta'}^*}}{\mathcal{G}^A_{\alpha' , {\alpha'}^*}}}}_{\mathfrak{L}^2(\mathcal{B}(L^2(\R^d)))} 
		\\
		&\qquad 
		= \sum_{(\alpha' , {\alpha'}^*) , (\beta' , {\beta'}^*) \in \Gamma \times \Gamma^*} \hat{g}^A_{\beta' , {\beta'}^* , \alpha' , {\alpha'}^*} \, {\Bscpro{ \sopro{\mathcal{G}^A_{\alpha , \alpha^*} }{\mathcal{G}^A_{\beta , \beta^*}} \, }{ \, \hat{F}^A \, \sopro{\mathcal{G}^A_{\beta' , {\beta'}^*}}{\mathcal{G}^A_{\alpha' , {\alpha'}^*}}}}_{\mathfrak{L}^2(\mathcal{B}(L^2(\R^d)))} 
		\\
		&\qquad 
		= \sum_{(\alpha' , {\alpha'}^*) , (\beta' , {\beta'}^*) \in \Gamma \times \Gamma^*} \hat{F}^A_{\alpha , \alpha^* , \beta' , {\beta'}^* , \alpha' , {\alpha'}^* , \beta , \beta^*} \, \hat{g}^A_{\beta' , {\beta'}^* , \alpha' , {\alpha'}^*} 
		.
	\end{align*}
	This completes the proof.
\end{proof}
Now we can furnish the proof of our main theorem. 
\begin{theorem}\label{boundedness_super_PsiDOs:thm:boundedness_super_PsiDOs}
	Assume the function $F \in S^m_{\rho,0}(\Pspace)$ is a Hörmander symbol of order $m \leq 0$ with $0 \leq \rho \leq 1$ and. Then $\Op^A(F)$ defines a bounded super operator $\Op^A(F) \in \mathcal{B} \bigl ( \mathfrak{L}^2 \bigl ( \mathcal{B} \bigl ( L^2(\R^d) \bigr ) \bigr ) \bigr )$ on the Hilbert space of Hilbert-Schmidt operators. 
\end{theorem}
\begin{proof}
	First of all, by the nesting property of Hörmander classes, 
	\begin{align}
		S^m_{\rho,0}(\Pspace) \subseteq S^m_{0,0}(\Pspace) 
		&&
		\forall \rho \in [0,1]
		, 
	\end{align}
	it suffices to verify the statement for $\rho = 0$.
	
	Let $F \in S_{0,0}^m(\Pspace)$ and $g \in \Schwartz(\pspace)$. Then it follows from Proposition~\ref{setting:prop:Hoermander_double_symbol_Schwartz_symbol_semi_super_product} that the semi-super product $F \semisuper^B g$ belongs to $\Schwartz(\pspace)$. Since $\op^A$ yields a unitary map from $L^2(\pspace)$ to $\mathfrak{L}^2 \bigl ( \mathcal{B} \bigl ( L^2(\R^d) \bigr ) \bigr )$ (\cf Subsection~\ref{setting:magnetic_PsiDOs}) and $\Schwartz(\pspace)\subset L^2(\pspace)$ we can deduce that $\Op^A(F) \, \hat{g}^A := \op^A \bigl ( F\semisuper^B g \bigr )$ belongs to $\mathfrak{L}^2 \bigl ( \mathcal{B} \bigl ( L^2(\R^d) \bigr ) \bigr )$.

	By using the Parseval identity and Proposition \ref{boundedness_super_PsiDOs:prop:superoperator_image_expansion} we obatin
	\begin{align}
		\bnorm{\Op^A(F) \, \hat{g}^A}_{\mathfrak{L}^2(\mathcal{B}(L^2(\R^d)))}^2 &= \sum_{(\alpha_{L,R}, \alpha_{L,R}^*) \in \Gamma\times\Gamma^*} \abs{ \bigl( \Op^A(F) \, \hat{g}^A \bigr)_{\alpha_L, \alpha_L^*, \alpha_R, \alpha_R^*} }^2 
		\notag \\
		&= \sum_{(\alpha_{L,R}, \alpha_{L,R}^*) \in \Gamma\times\Gamma^*} \biggl| \sum_{(\beta_{L,R}, \beta_{L,R}^*) \in \Gamma\times\Gamma^*} \Op^A(F)_{\alpha_L, \alpha_L^*, \beta_L, \beta_L^*, \beta_R, \beta_R^*, \alpha_R, \alpha_R^*} \, \hat{g}^A_{\beta_L, \beta_L^*, \beta_R, \beta_R^*} \biggr|^2 
		.
		\label{boundedness_super_PsiDOs:eqn:applying_Parseval}
	\end{align}
	It follows from Schur's Test (Lemma \ref{boundedness_super_PsiDOs:lem:Schur}) and Proposition \ref{abstract_matrix_representation:prop:Hilbert_Schmidt} that we have
	\begin{align*}
		\sum_{(\alpha_{L,R}, \alpha_{L,R}^*) \in \Gamma\times\Gamma^*} &\abs{ \sum_{(\beta_{L,R}, \beta_{L,R}^*) \in \Gamma\times\Gamma^*} \Op^A(F)_{\alpha_L, \alpha_L^*, \beta_L, \beta_L^*, \beta_R, \beta_R^*, \alpha_R, \alpha_R^*} \, \hat{g}^A_{\beta_L, \beta_L^*, \beta_R, \beta_R^*} }^2 
		\\
		&\leq C(F,B)^2 \sum_{(\beta_{L,R}, \beta_{L,R}^*) \in \Gamma\times\Gamma^*} \abs{ \hat{g}^A_{\beta_L, \beta_L^*, \beta_R, \beta_R^*} }^2 = C(F,B)^2 \norm{\hat{g}^A}_{\mathfrak{L}^2 (\mathcal{B} ( L^2(\R^d) ) )}^2 
		.
	\end{align*}
	Combining this with \eqref{boundedness_super_PsiDOs:eqn:applying_Parseval} we get
	\begin{align*}
		\bnorm{\Op^A(F) \, \hat{g}^A}_{\mathfrak{L}^2(\mathcal{B}(L^2(\R^d)))} \leq C(F,B) \norm{\hat{g}^A}_{\mathfrak{L}^2 (\mathcal{B} ( L^2(\R^d) ) )} \qquad \forall \hat{g}^A \in \op^A \bigl ( \Schwartz(\pspace)  \bigr )
		.
	\end{align*}
	The magnetic Weyl quantization map $\op^A$ yields a unitary map from $L^2(\pspace)$ to $\mathfrak{L}^2 \bigl ( \mathcal{B} \bigl ( L^2(\R^d) \bigr ) \bigr )$ (\cf Subsection \ref{setting:magnetic_PsiDOs}) and $\Schwartz(\pspace)$ is dense in $L^2(\pspace)$. This shows that $\op^A \bigl ( \Schwartz(\pspace) \bigr ) := \bigl \{ \op^A(g) \; \; \vert \; \; g \in \Schwartz(\pspace) \bigr \}$ is a dense subspace of $\mathfrak{L}^2 \bigl ( \mathcal{B} \bigl ( L^2(\R^d) \bigr ) \bigr )$, and hence $\Op^A(F)$ uniquely extends to a bounded linear operator from $\mathfrak{L}^2 \bigl ( \mathcal{B} \bigl ( L^2(\R^d) \bigr ) \bigr )$ to itself. The proof is complete.
\end{proof}
\begin{remark}
	If we \emph{could} prove that 
	\begin{align*}
		\Op^A(F) \in \mathcal{B} \bigl ( \mathcal{B} \bigl ( L^2(\R^d) \bigr ) \bigr ) = \mathcal{B} \Bigl ( \mathfrak{L}^{\infty} \bigl ( \mathcal{B} \bigl ( L^2(\R^d) \bigr ) \bigr ) \Bigr )
	\end{align*}
	was bounded on $\mathcal{B} \bigl ( L^2(\R^d) \bigr )$ under the \emph{same} conditions as in Theorem~\ref{boundedness_super_PsiDOs:thm:boundedness_super_PsiDOs}, the non-commutative version of the Riesz-Thorin Interpolation Theorem (\cf \eg Theorems~2.9 and 2.10 as well as Remark~1 on p.~23 in \cite{Simon:trace_ideals_applications:2005} or \cite[Theorem~13.1]{Gohberg_Krein:linear_nonselfadjoint_operators:1969}) would imply that 
	\begin{align*}
		\Op^A(F) : \mathfrak{L}^p \bigl ( \mathcal{B} \bigl ( L^2(\R^d) \bigr ) \bigr ) \longrightarrow \mathfrak{L}^p \bigl ( \mathcal{B} \bigl ( L^2(\R^d) \bigr ) \bigr ) 
	\end{align*}
	defined a bounded operator for \emph{any} $2 \leq p \leq \infty$. 
	
	Likewise, a proof of boundedness on the trace-class operators $\mathfrak{L}^1 \bigl ( \mathcal{B} \bigl ( L^2(\R^d) \bigr ) \bigr )$ would imply boundedness for any $1 \leq p \leq 2$. 
\end{remark}
\begin{remark}[Extension to operator-valued symbols]
	In many applications, one encounters pseudodifferential operators defined from \emph{operator-valued} functions (see \eg \cite{Cordes:pseudodifferential_FW_transform:1983,PST:sapt:2002,PST:Born-Oppenheimer:2007,Fuerst_Lein:scaling_limits_Dirac:2008}); a systematic construction of the operator-valued calculus from first principles can be found in \cite{DeNittis_Lein_Seri:equivariant_PsiDOs:2022}. Combining the ideas from \cite{DeNittis_Lein_Seri:equivariant_PsiDOs:2022} and \cite{Lee_Lein:magnetic_pseudodifferential_super_operators_basics:2022} should give a pseudodifferential \emph{super} calculus defined for operator-valued and equivariant operator-valued Hörmander symbols. 
	
	Specifically, the ideas from \cite[Section~3.4.3]{DeNittis_Lein_Seri:equivariant_PsiDOs:2022} should apply verbatim and would allow for an extension of our boundedness result, Theorem~\ref{boundedness_super_PsiDOs:thm:boundedness_super_PsiDOs}, to the operator-valued context. The abstract results from Section~\ref{abstract_matrix_representation} apply verbatim to \eg $L^2(\R^d,\Hil)$, where $\Hil$ is some separable Hilbert space. The main technical result, Theorem~\ref{characterization_symbols_matrix_elements_super_PsiDOs:thm:characterization_Hoermander_class_super_PsiDOs}, should also extend, following the strategy and arguments in \cite{DeNittis_Lein_Seri:equivariant_PsiDOs:2022}. 
\end{remark}
The inclusion $S_{\rho,0}^{m_L,m_R}(\Pspace) \subseteq S_{0,0}^0(\Pspace)$ for $m_L , m_R \leq 0$ allows us to apply Theorem~\ref{boundedness_super_PsiDOs:thm:boundedness_super_PsiDOs} also to double Hörmander classes from Definition~\ref{setting:defn:double_Hoermander_symbol_classes} and obtain an boundedness result for them. 
\begin{corollary}\label{boundedness_super_PsiDOs:cor:boundedness_super_PsiDOs}
	Suppose that $F \in S_{\rho,0}^{m_L,m_R}(\Pspace)$, $m_L , m_R \leq 0$, with $0 \leq \rho \leq 1$. Then $\Op^A(F)$ defines a bounded super operator $\Op^A(F) \in \mathcal{B} \bigl ( \mathfrak{L}^2 \bigl ( \mathcal{B} \bigl ( L^2(\R^d) \bigr ) \bigr ) \bigr )$ on the Hilbert space of Hilbert-Schmidt operators.
\end{corollary}
%
% subsection boundedness_of_magnetic_pseudodifferential_superoperators (end)

\subsection{More advanced boundedness results} % (fold)
\label{boundedness_super_PsiDOs:advanced_boundedness_results}
Let us give a brief outlook on how to obtain more advanced boundedness results for super operators from Theorem~\ref{boundedness_super_PsiDOs:thm:boundedness_super_PsiDOs} and Corollary~\ref{boundedness_super_PsiDOs:cor:boundedness_super_PsiDOs}.

One that should be useful in applications is an extension from $\mathfrak{L}^2$ to non-commutative magnetic Sobolev spaces. We first revisit the case of magnetic pseudodifferential operators: if $f \in S^m_{\rho,0}(\pspace)$ is a Hörmander symbol of order $m$, then the associated magnetic seudodifferential operator $\op^A(f)$ defines a bounded operator 
\begin{align*}
	\op^A(f) : H^s_A(\R^d) \longrightarrow H^{s - m}_A(\R^d)
\end{align*}
between magnetic Sobolev spaces for any $s \in \R$ (\cf \cite[Proposition~4.12]{Iftimie_Mantoiu_Purice:magnetic_psido:2006}). Rather than use the standard magnetic Sobolev norm, we instead use an equivalent norm defined from 
\begin{align*}
	w_m(x,\xi) := \sexpval{\xi}^m + \lambda(m)
\end{align*}
for $m \geq 0$. To extend this definition to negative orders, we choose the constant $\lambda(m) > 0$ large enough so that the inverse 
\begin{align*}
	w_{-m} = \bigl ( \sexpval{\xi}^m + \lambda(m) \bigr )^{(-1)_{\star^B}} 
\end{align*}
with respect to the magnetic Weyl product $\star^B$ exists (\cf \cite[Proposition~6.31]{Iftimie_Mantoiu_Purice:commutator_criteria:2008} or \cite[Theorem~1.8]{Mantoiu_Purice_Richard:Cstar_algebraic_framework:2007}). Importantly, $w_m \in S^m_{1,0}(\pspace)$ is a Hörmander symbol of order $m \in \R$ and $\bnorm{\op^A(w_m) \psi}_{L^2(\R^d)}$ is an equivalent norm for the Banach space $H^m_A(\R^d)$. Consequently, we have proven the following little 
\begin{lemma}\label{outlook:lem:characterization_bounded_operators_magnetic_Sobolev_spaces}
	$\hat{f}^A \in \mathcal{B} \bigl ( H^m_A(\R^d) \, , \, H^s_A(\R^d) \bigr )$ is a bounded operator between magnetic Sobolev spaces of orders $m \in \R$ and $s \in \R$ if and only if 
	\begin{align*}
		\Op^A \bigl ( w_{-s} \otimes w_{-m} \bigr ) \, \hat{f}^A &= \op^A(w_{-s}) \, \hat{f}^A \, \op^A(w_{-m}) \in \mathcal{B} \bigl ( L^2(\R^d) \bigr ) 
	\end{align*}
	defines a bounded operator on $L^2(\R^d)$. 
\end{lemma}
The other piece we will need is a notion of $p$-Schatten classes $\mathfrak{L}^p \bigl ( \mathcal{B}(\Hil_1,\Hil_2) \bigr )$ for operators that map between two different separable Hilbert spaces. While the ideas are all contained in \cite{Treves:topological_vector_spaces:1967}, a short, guided introduction to the topic can be found in \cite[Appendix~B]{DeNittis_Lein_Seri:equivariant_PsiDOs:2022}. For the present purposes it suffices to note that we may set 
\begin{align*}
	\mathfrak{L}^p \bigl ( \mathcal{B}(\Hil_1,\Hil_2) \bigr ) := \Bigl \{ \hat{f} \in \mathcal{B}(\Hil_1,\Hil_2) \; \; \big \vert \; \; \babs{\hat{f}} \in \mathfrak{L}^p \bigl ( \mathcal{B}(\Hil_1) \bigr ) \Bigr \} 
\end{align*}
where we have defined the absolute value as $\babs{\hat{f}} := \sqrt{\hat{f}^* \, \hat{f}}$. Equivalently, we could have put the adjoint operator to the right and base our definition on the absolute value $\babs{\hat{f}} := \sqrt{\hat{f} \, \hat{f}^*}$. 

That being said, we can now exploit Lemma~\ref{outlook:lem:characterization_bounded_operators_magnetic_Sobolev_spaces} to derive boundedness results for magnetic pseudodifferential super operators. 
\begin{corollary}\label{boundedness_super_PsiDOs:cor:boundedness_super_PsiDOs_magnetic_Sobolev_spaces_double_Hoermander_symbols}
	Suppose $F \in S^{m_L,m_R}_{\rho,0}(\Pspace)$ is a double Hörmander symbol of orders $m_L , m_R \in \R$ with $0 \leq \rho \leq 1$. Then $\Op^A(F)$ defines a bounded operator 
	\begin{align}
		\Op^A(F) \in \mathcal{B} \Bigl ( \mathfrak{L}^2 \bigl ( H^{m_{\mathrm{i},R}}_A(\R^d) \, , \, H^{m_{\mathrm{i},L}}_A(\R^d) \bigr ) \, , \, \mathfrak{L}^2 \bigl ( H^{m_{\mathrm{f},R}}_A(\R^d) \, , \, H^{m_{\mathrm{f},L}}_A(\R^d) \bigr ) \Bigr )
		\label{outlook:eqn:bounded_Psi_super_DO_magnetic_Sobolev_spaces_double_Hoermander_symbols}
	\end{align}
	whenever the coefficients satisfy 
	\begin{align*}
		m_L &\leq m_{\mathrm{i},L} + m_{\mathrm{f},L}
		, 
		\\
		m_R &\leq m_{\mathrm{i},R} + m_{\mathrm{f},R}
		. 
	\end{align*}
\end{corollary}
\begin{proof}
	First of all, with the help of Lemma~\ref{outlook:lem:characterization_bounded_operators_magnetic_Sobolev_spaces} we can reformulate equation~\eqref{outlook:eqn:bounded_Psi_super_DO_magnetic_Sobolev_spaces_double_Hoermander_symbols} as 
	\begin{align*}
		&\Op^A \bigl ( w_{-m_{\mathrm{f},L}} \otimes w_{-m_{\mathrm{f},R}} \bigr ) \, \Op^A(F) \, \Op^A \bigl ( w_{-m_{\mathrm{i},L}} \otimes w_{-m_{\mathrm{i},R}} \bigr ) 
		= \\
		&\qquad \qquad 
		= \Op^A \Bigl ( \bigl ( w_{-m_{\mathrm{f},L}} \otimes w_{-m_{\mathrm{f},R}} \bigr ) \sharp^B F \sharp^B \bigl ( w_{-m_{\mathrm{i},L}} \otimes w_{-m_{\mathrm{i},R}} \bigr ) \Bigr ) 
		\in \mathcal{B} \Bigl ( \mathfrak{L}^p \bigl ( L^2(\R^d) \bigr ) \Bigr ) 
	\end{align*}
	being a bounded operator on the standard $p$-Schatten class for the Hilbert space $L^2(\R^d)$. 
	
	We know from \cite[Proposition~VI.5~(2)]{Lee_Lein:magnetic_pseudodifferential_super_operators_basics:2022} that the right-hand side is a magnetic pseudodifferential operator defined by a symbol from the double symbol class of orders $\bigl ( m_L - m_{\mathrm{i},L} - m_{\mathrm{f},L} \, , \, m_R - m_{\mathrm{i},R} - m_{\mathrm{f},R} \bigr )$. Combining this with Corollary~\ref{boundedness_super_PsiDOs:cor:boundedness_super_PsiDOs} yields the conditions 
	\begin{align*}
		m_L - m_{\mathrm{i},L} - m_{\mathrm{f},L} &\leq 0 
		, 
		\\
		m_R - m_{\mathrm{i},R} - m_{\mathrm{f},R} &\leq 0 
		, 
	\end{align*}
	which are equivalent to the ones stated in the proof. 
\end{proof}
By the same token, we get a similar statement for regular Hörmander classes on $\Pspace$ via Theorem~\ref{boundedness_super_PsiDOs:thm:boundedness_super_PsiDOs}: 
\begin{corollary}\label{boundedness_super_PsiDOs:cor:boundedness_super_PsiDOs_magnetic_Sobolev_spaces_regular_Hoermander_symbols}
	Suppose $F \in S^m_{\rho,0}(\Pspace)$ is a double Hörmander symbol of orders $m\in \R$ with $0 \leq \rho \leq 1$. Then $\Op^A(F)$ defines a bounded operator 
	\begin{align}
		\Op^A(F) \in \mathcal{B} \Bigl ( \mathfrak{L}^2 \bigl ( H^{m_{\mathrm{i},R}}_A(\R^d) \, , \, H^{m_{\mathrm{i},L}}_A(\R^d) \bigr ) \, , \, \mathfrak{L}^2 \bigl ( H^{m_{\mathrm{f},R}}_A(\R^d) \, , \, H^{m_{\mathrm{f},L}}_A(\R^d) \bigr ) \Bigr )
		\label{outlook:eqn:bounded_Psi_super_DO_magnetic_Sobolev_spaces_double_Hoermander_symbols_regular_Hoermander_symbols}
	\end{align}
	whenever the coefficients satisfy 
	\begin{align*}
		m &\leq m_{\mathrm{i},L} + m_{\mathrm{f},L}
		, 
		\\
		m &\leq m_{\mathrm{i},R} + m_{\mathrm{f},R}
		. 
	\end{align*}
\end{corollary}
\begin{remark}
	The limitation to the case $p = 2$ is entirely due to the absence of boundedness results akin to Theorem~\ref{boundedness_super_PsiDOs:thm:boundedness_super_PsiDOs} and Corollary~\ref{boundedness_super_PsiDOs:cor:boundedness_super_PsiDOs} for $p \neq 2$. However, \emph{if} one can prove that 
	\begin{align*}
		\Op^A(F) \in \mathcal{B} \bigl ( \mathfrak{L}^p \bigl ( \mathcal{B} \bigl ( L^2(\R^d) \bigr ) \bigr ) \bigr )
	\end{align*}
	holds true for some $F \in S^{0,0}_{0,0}(\Pspace)$ or $F \in S^0_{0,0}(\Pspace)$ for $p \neq 2$, then the above arguments immediately imply boundedness in the sense that 
	\begin{align*}
		\Op^A(F) \in \mathcal{B} \Bigl ( \mathfrak{L}^p \bigl ( H^{m_{\mathrm{i},R}}_A(\R^d) \, , \, H^{m_{\mathrm{i},L}}_A(\R^d) \bigr ) \, , \, \mathfrak{L}^p \bigl ( H^{m_{\mathrm{f},R}}_A(\R^d) \, , \, H^{m_{\mathrm{f},L}}_A(\R^d) \bigr ) \Bigr )
	\end{align*}
	where the $m_{\mathrm{i},L/R}$ and $m_{\mathrm{f},L/R}$ satisfy the same conditions as in Corollaries~\ref{boundedness_super_PsiDOs:cor:boundedness_super_PsiDOs_magnetic_Sobolev_spaces_double_Hoermander_symbols} or \ref{boundedness_super_PsiDOs:cor:boundedness_super_PsiDOs_magnetic_Sobolev_spaces_regular_Hoermander_symbols} above. 
\end{remark}
%
% subsection more_advanced_boundedness_results (end)
% section boundedness_of_magnetic_pseudodifferential_super_operators (end)
% \include{./counterexamples}
% \include{./section_6}
% \include{./appendix}
%%% end content %%% (end)

% \printbibliography
\bibliographystyle{alpha}
\bibliography{bibliography}

\end{document}